%% file: main.tex
\documentclass[11pt,envcountsame]{llncs}
\spnewtheorem{observation}[theorem]{Observation}{\bfseries}{\itshape}

\usepackage{microtype}
\usepackage[utf8]{inputenc}
\usepackage{lmodern}
\usepackage[T1]{fontenc}
\usepackage{textcomp}
\usepackage[mathscr]{eucal}
\usepackage{amssymb}
\usepackage{soul}
\usepackage{color}
\usepackage[USenglish]{babel}
\usepackage[tbtags,fleqn]{amsmath}
\usepackage{enumerate}
\usepackage{graphicx}
\usepackage{array}
\usepackage{multirow}
\usepackage{tabularx}
\usepackage[online]{threeparttable}
\usepackage{listings}
\usepackage{lastpage}
\usepackage[hidelinks]{hyperref}
\usepackage[style=base,labelsep=space,singlelinecheck=false,font={up,small},labelfont={sf,bf}]{caption}
\usepackage[figuresright]{rotating}
\usepackage{subcaption}
\usepackage{comment}
\usepackage{xstring}

\usepackage{amsthm}
\usepackage[capitalise,noabbrev]{cleveref}
\usepackage{aliascnt}
\usepackage[left,mathlines]{lineno}

\usepackage[dvipsnames]{xcolor}
\usepackage[textsize=scriptsize]{todonotes}
\usepackage{tabto}
\usepackage{dsfont}
\usepackage{cite}
\usepackage{marvosym}
\usepackage{mathtools}
\usepackage{xspace}
\usepackage[ruled,vlined,linesnumbered]{algorithm2e}
\usepackage{pgffor}
\usepackage{contour}

\DeclareFontFamily{U}{stmry}{}
\DeclareFontShape{U}{stmry}{m}{n}
{ <-6> stmary5 <6-7> stmary6 <7-8> stmary8
	<8-9> stmary9 <9-10> stmary9
	<10-> stmary10
}{}
\usepackage{stmaryrd}
\SetSymbolFont{stmry}{bold}{U}{stmry}{m}{n}

\usepackage{tikz}
\usetikzlibrary{backgrounds}
\usetikzlibrary{calc}
\usetikzlibrary{math}
\usetikzlibrary{positioning}
\usetikzlibrary{arrows.meta}
\usetikzlibrary{decorations.pathreplacing}
\usetikzlibrary{decorations.pathmorphing}
\usetikzlibrary{shapes.misc}
\usetikzlibrary{patterns}
\definecolor{superred}{rgb}{0.87,0.1,0.1}
\definecolor{superblue}{RGB}{66,102,170}

\let\oldsum\sum
\renewcommand{\sum}{{\textstyle\oldsum\nolimits}}
\let\oldbigcup\bigcup
\renewcommand{\bigcup}{{\textstyle\oldbigcup\nolimits}}
\let\oldbigcap\bigcap
\renewcommand{\bigcap}{{\textstyle\oldbigcap\nolimits}}
\renewcommand{\frac}[2]{{\tfrac{#1}{#2}}}
\let\oldexists\exists
\renewcommand{\exists}{\oldexists\,}
\let\oldnexists\nexists
\renewcommand{\nexists}{\oldnexists\,}
\let\oldforall\forall
\renewcommand{\forall}{\oldforall\,}

\newcommand{\1}[1]{\operatorname{\mathds{1}}\!\big( #1 \big)}

\newcommand{\MaxMPM}{\textsf{Max-MPM}\xspace}
\newcommand{\MinMPM}{\textsf{Min-MPM}\xspace}
\newcommand{\MIM}{\textsf{MIM}\xspace}
\newcommand{\MUM}{\textsf{MUM}\xspace}
\newcommand{\IM}[1]{\textsf{#1-IM}\xspace}
\newcommand{\UM}[1]{\textsf{#1-UM}\xspace}
\newcommand{\opt}{\ensuremath{\textsf{opt}}}
\newcommand{\apx}{\ensuremath{\textsf{apx}}}
\newcommand{\profit}{p}
\newcommand{\uprofit}{c}
\newcommand{\qceil}{\overline{q}}
\DeclarePairedDelimiter{\range}{[}{]}
\DeclarePairedDelimiter{\Range}{\llbracket}{\rrbracket}
\newcommand{\set}[2]{\pgfmathsetmacro{#1}{int(#2)}}

\newcommand{\G}{\ensuremath{\mathcal{G}}\xspace}
\newcommand{\M}{\ensuremath{\mathcal{M}}\xspace}

\SetKwComment{Comment}{// }{}
\DontPrintSemicolon

\begin{document}
	\title{Approximating Multistage Matching Problems}
	\titlerunning{Approx.\ Multistage Matching Problems}
	\author{Markus Chimani\orcidID{0000-0002-4681-5550}
		\and Niklas Troost\orcidID{0000-0001-7412-2770} 
		\and Tilo Wiedera\orcidID{0000-0002-5923-4114}}
	\authorrunning{M. Chimani\and N. Troost\and T. Wiedera}
	\institute{Theoretical Computer Science, Osnabrück University\\
		\email{$\{\text{markus.chimani,niklas.troost,tilo.wiedera}\}$@uos.de}}
	\maketitle

\begin{abstract}
In \emph{multistage perfect matching} problems, we are given a sequence of graphs on the same vertex set and are asked to find a sequence of perfect matchings, corresponding to the sequence of graphs, such that consecutive matchings are as similar as possible.
More precisely, we aim to maximize the intersections, or minimize the unions between consecutive matchings.

We show that these problems are NP-hard even in very restricted scenarios.
As our main contribution, we present the first non-trivial approximation algorithms for these problems:
On the one hand, we devise a tight approximation on graph sequences of length~two ($2$-stage graphs).
On the other hand, we propose several general methods to deduce multistage approximations from blackbox approximations on $2$-stage graphs.
\end{abstract}

	\section{Introduction}
	The study of graphs that evolve over time emerges naturally in several applications. As such, it is a well-known subject in graph theory~\cite{AMSZ20,BEK19,BEST19,BET19,BELP18,BBR19,BS15,BHI18,BLSZ14,C18,E91,FNRZ19,GTW14,HHKNRS19,T00,KKK00,OS14}.
	While there are many possible approaches to model these problems (cf.\ the discussion of related work), the paradigm of \emph{multistage graphs} has attracted quite some attention in recent years \cite{BEK19,BELP18,BEST19,BET19,FNRZ19,GTW14,HHKNRS19}.
	In this setting, we are given a sequence of separate, but related graphs (\emph{stages}).
	A typical goal is to find a sequence of solutions for each individual graph such that the change in the solutions between consecutive graphs is minimized.
	Since multistage graph problems usually turn out to be NP-hard, one often resorts to FPT- or approximation algorithms.
	To the best of our knowledge, all approximation results in this setting discuss combined objective functions that reflect a trade-off between the quality of each individual solution and the cost of the change over time (cf., e.g., \cite{GTW14,BELP18}).
	However, this is a drawback if one requires each stage's solution to attain a certain quality guarantee (e.g., optimality).
	Trying to ensure this by adjusting the trade-off weights in the above approximation algorithms leads to approximation ratios that no longer effectively bound the cost of change.
	Here, we discuss a multistage graph problem where each individual solution is necessarily optimal, but we can still obtain an approximation ratio on the cost of the change over time.

	A classical example are \emph{multistage matching} problems, i.e., natural multistage generalizations of traditional matching problems (e.g., perfect matching, maximum weight matching, etc.).
	This is particularly interesting as optimality for a single stage would be obtainable in polynomial time, but all known multistage variants are NP-hard already for two stages.
	There are several known approximation algorithms for multistage matching problems~\cite{BELP18};
		however, they all follow the trade-off paradigm.

	In this paper, we are concerned with maintaining a perfect matching on a multistage graph,
		such that the changes between consecutive matchings are minimized.
	After showcasing the complexities of our problems (\cref{sec:compl}),
		we will devise efficient approximation algorithms (\cref{sec:apx}).

	\paragraph{Definitions and preliminaries.}
	Let $G=(V,E)$ be an undirected graph.
	For a set~$W\subseteq V$ of vertices, let $\delta(W) \coloneqq \big\{uv\in E\mid u\in W, v\in V\setminus W\big\}$ denote the set of its cut edges.
	For a singleton $\{v\}$, we may write $\delta(v)$ instead of $\delta(\{v\})$.
	A set~$M\subseteq E$ of edges is a \emph{matching}
		if every vertex is incident to at most one edge of $M$;
		it is a \emph{perfect} matching if $|\delta(v)\cap M|=1$ for every $v\in V$.
	A $k$-\emph{cycle} ($k$-\emph{path}) is a cycle (path, respectively) consisting of exactly $k$ edges.
	The parity of a $k$-cycle is the parity of~$k$.
	For a set~$F\subseteq E$ of edges, let $V(F) \coloneqq \{v\in V\mid \delta(v)\cap F\neq \varnothing\}$ denote its incident vertices.

	For $x\in\mathbb N$, we define $\range{x} \coloneqq \{1,...,x\}$ and $\Range{x} \coloneqq \{0\}\cup\range{x}$.
	A \emph{temporal graph} (or \emph{$\tau$-stage graph}) is a tuple $\G=(V,E_1,..., E_\tau)$
		consisting of a vertex set~$V$ and multiple edge sets~$E_i$, one for each $i\in\range{\tau}$.
	The graph $G_i \coloneqq (V(E_i),E_i)$ is the $i$\emph{th stage of}~\G.
	We define $n_i \coloneqq |V(E_i)|$, and $n \coloneqq  |V|$.
	A temporal graph is \emph{spanning} if $V(E_i) = V$ for each $i\in\range{\tau}$.

	Let $\mu \coloneqq \max_{i\in\range{\tau -1}} |E_i\cap E_{i+1}|$ denote the maximum number of edges that are common between two adjacent stages.
	Let $E_\cap \coloneqq \bigcap_{i\in\range{\tau}} E_i$ and $E_\cup \coloneqq \bigcup_{i\in\range{\tau}} E_i$.
	The graph $G_\cup \coloneqq(V(E_\cup),E_\cup)$ is the \emph{union graph} of~$\G$.
	A~\emph{multistage perfect matching} in~\G is a sequence $\M \coloneqq (M_i)_{i\in\range{\tau}}$ such that for each~$i\in\range{\tau}$, $M_i$~is a perfect matching in $G_i$.

	All problems considered in this paper (\MIM, \MUM, \MinMPM, \MaxMPM; see below) are of the following form:
	Given a temporal graph~$\G$, we ask for a multistage perfect matching $\M$ optimizing some objective function.
	In their respective decision variants, the input furthermore consists of some value $\kappa$
	and we ask whether there is an \M with objective value at most (minimization problems) or at least (maximization problems)~$\kappa$.

	\begin{definition}[\MIM and $\IM{$\tau$}$]
		Given a temporal graph~$\G$, the \emph{multistage intersection matching problem~(\MIM)} asks for a multistage perfect matching~$\M$ of~$\G$ with maximum
		\emph{profit}~$\profit(\M) \coloneqq \sum_{i\in\range{\tau-1}} |M_i\cap M_{i+1}|$.
		For fixed~$\tau$, we denote the problem by~$\IM{$\tau$}$.
	\end{definition}

	We also consider the natural inverse objective, i.e., minimizing the unions.
		While the problems differ in the precise objective function, an optimal solution of \MIM is optimal for \MUM as well, and vice versa.

	\begin{definition}[\MUM and $\UM{$\tau$}$]
		Given a temporal graph~$\G$, the \emph{multistage union matching problem~(\MUM)} asks for a multistage perfect matching~$\M$ of~$\G$ with minimum~\emph{cost}~$\uprofit(\M) \coloneqq \sum_{i\in\range{\tau-1}} |M_i\cup M_{i+1}|$.
		For fixed~$\tau$, we denote the problem by~$\UM{$\tau$}$.
	\end{definition}

	Consider either \MIM or \MUM.
	Given a temporal graph $\G$, we denote with $\opt$ the optimal solution value and with $\apx$ the objective value achieved by a given algorithm with input $\G$.
	The \emph{approximation ratio} of an approximation algorithm for \MIM (\MUM) is the infimum (supremum, respectively) of $\apx / \opt$ over all instances.

	\paragraph{Related work.}
	The classical \emph{dynamic graph} setting often considers small modifications, e.g., single edge insertions/deletions~\cite{E91,T00}.
	There, one is given a graph with a sequence of modifications
		and asked for a feasible solution after each modification.
	A natural approach to tackle matchings in such graphs is to make local changes to the previous solutions~\cite{BS15,BHI18,BLSZ14,S07}.

	A more general way of modeling changes is that of \emph{temporal graphs},
		introduced by Kempe et al.~\cite{KKK00} and used herein.
	Typically, each vertex and edge is assigned a set of time intervals that specify when it is present.
	This allows an arbitrary number of changes to occur at the same time.
	Algorithms for this setting usually require a more global perspective and many approaches do not rely solely on local changes.
	In fact, many temporal (matching) problems turn out to be hard, even w.r.t.\ approximation and fixed-parameter-tractability~\cite{AMSZ20, BBR19, C18, MMNZZ19, OS14}.

	One particular flavor of temporal graph problems is concerned with obtaining a sequence of solutions---one for each stage---while optimizing a global quantity.
	These problems are often referred to as \emph{multistage problems} and gained much attention in recent years~\cite{BELP18, BEK19, BEST19, BET19, FNRZ19, GTW14, HHKNRS19},
		including in the realm of matchings:
		e.g., the authors of~\cite{HHKNRS19} show $\mathrm W[1]$-hardness for finding the largest edge set that induces a matching in each stage.

	In the literature we find the problem \MaxMPM, where the graph is augmented with time-dependent edge weights, and we want to maximize
	the value of each individual perfect matching (subject to the given edge costs) \emph{plus} the total profit~\cite{BELP18}.
	\MIM is the special case where all edge costs are zero, i.e., we only care about the multistage properties of the solution, as long as each stage is perfectly matched.
	There is also the inverse optimization problem \MinMPM, where we minimize the value of each perfect matching \emph{plus} the number of matching edges newly introduced in each stage.
	We have APX-hardness for \MaxMPM and \MinMPM~\cite{BELP18,GTW14}
		(for \MinMPM one may assume a complete graph at each stage, possibly including edges of infinite weight).
	The latter remains APX-hard
	even for spanning 2-stage graphs with bipartite union graph and no edge weights
	(i.e., we only minimize the number of edge swaps)~\cite{BELP18}.
	For uniform edge weights~$0$, the objective of \MinMPM is to minimize $\sum_{i\in\range{\tau -1}} |M_{i+1} \setminus M_i|$;
		similar but slightly different to $\MUM$ (equal up to additive $\sum_{i\in\range{\tau - 1}} n_i/2$).
	For \MinMPM on metric spanning 2- or 3-stage graphs, the authors of~\cite{BELP18} show $3$-approximations.
	They also
		propose a $(1/2)$-approximation for \MaxMPM on spanning temporal
		graphs with any number of stages, which is unfortunately wrong (see Appendix~\ref{appendix:MaxMPM} for a detailed discussion).

	When restricting
		\MaxMPM and \MinMPM to uniform edge weights~$0$,
		optimal solutions for \MIM, \MUM, \MaxMPM, and \MinMPM are identical; thus \MIM and \MUM are NP-hard.
	However, the APX-hardness of \MinMPM does not imply APX-hardness of \MUM
		as the objective functions slightly differ.
	Furthermore, the APX-hardness reduction to \MaxMPM inherently requires non-uniform
		edge weights and does not translate to \MIM.
	To the best of our knowledge, there are no non-trivial approximation algorithms for any of these problems on more than three stages.

	\paragraph{Our contribution.}
	We start with showing in \cref{sec:compl} that the problems are NP-hard even in much more restricted scenarios than previously known,
		and that (a lower bound for) the integrality gap of the natural linear program for \IM 2
		is close to the approximation ratio we will subsequently devise.
	This hints that stronger approximation ratios may be hard to obtain, at least using LP techniques.

	As our main contribution, we propose several approximation algorithms for the multistage problems \MIM and \MUM,
		as well as for their stage-restricted variants, see Fig.~\ref{fig:overview}.
	In particular, in Section~\ref{sec:2im-apx},
		we show a $(1/\!\sqrt{2\mu})$-approximation for \IM 2 and that this analysis is tight.
	Then, in Section~\ref{sec:mim-apx},
		we show that any approximation of \IM 2
		can be used to derive two different approximation algorithms for \MIM,
		whose approximation ratios are a priori incomparable.
	In Section~\ref{sec:mum-apx},
		we further show how to use all these algorithms to approximate \MUM (and \UM 2).
	We also observe that it is infeasible to use an arbitrary \MUM algorithm to approximate \MIM.
	In particular, we propose the seemingly first approximation algorithms for \MIM and \MUM on arbitrarily many stages.
	We stress that our goal is to always guarantee a perfect matching in each stage;
		the approximation ratio deals purely with optimizing the transition costs.
	Recall that approximation algorithms optimizing a weighted sum between intra- and interstage costs cannot guarantee such solutions in general.

	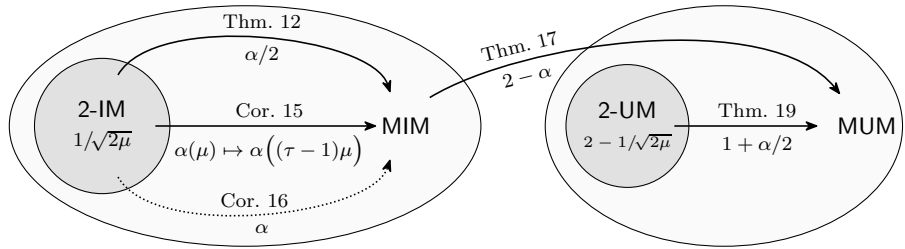
\begin{figure}[tb]
		\centering
		\begin{tikzpicture}[xscale=.95]
			\draw[fill=black!2] (0,0) ellipse[x radius=33mm, y radius=16mm];
			\draw[fill=black!2] (6.7,0) ellipse[x radius=25mm, y radius=16mm];

			\node[minimum width=6mm,circle] (mim) at (2.25,0) {\MIM};
			\node[minimum width=4mm,circle] (mum) at (8.7,0) {\MUM};

			\node[align=center,circle,draw,minimum width=18mm,fill=black!12] (tim) at (-2,0) {\IM 2\\ \smaller $1/\!\sqrt{2\mu}$};
			\node[align=center,circle,draw,minimum width=15mm,fill=black!12] (tum) at (5.35,0) {\UM 2\\ \smaller\smaller $2-1/\!\sqrt{2\mu}$};

			\coordinate (x) at (1.8,1.2);
			\coordinate (y) at (7,1.2);

			\tikzset{every edge/.style={draw,-Stealth[round],semithick}}
			\path
				(tim.70) edge[bend left=65,distance=8mm,shorten <=-2mm] node[below] {\smaller$\alpha / 2$} node[above] {\smaller Thm.~\ref{thm:2IM-to-MIM}} (mim.115)
				(tim) edge[bend right=0,shorten <=-2mm,shorten >=-1mm] node[below] {\smaller$\alpha(\mu) \mapsto \alpha\big((\tau-1)\mu\big)$} node[above] {\smaller Cor.~\ref{cor:better-tim}} (mim)
				(tim.-70) edge[bend right=65,distance=8mm,densely dotted,shorten <=-2mm] node[below] {\smaller $\alpha$} node[above,xshift=-1mm] {\smaller Cor.~\ref{cor:const}} (mim.-115)
				(mim.50) edge[out=30,in=130,distance=15mm] node[below,sloped,pos=0.25] {\smaller $2-\alpha$} node[above,sloped,pos=0.25] {\smaller Thm.~\ref{thm:MIM-to-MUM}} (mum.130)
				(tum) edge[bend right=0,shorten <=-2mm,shorten >=1mm] node[below] {\smaller $1+\alpha / 2$} node[above] {\smaller Thm.~\ref{thm:2UM-to-MUM}} (mum)
			;
		\end{tikzpicture}
		\caption{Relations of our approximation results.\label{fig:overview}
		An arc from problem~$A$ to~$B$ labeled~$f(\alpha)$ denotes the existence of an~$f(\alpha)$-approximation for~$B$, given an $\alpha$-approximation for~$A$.
		In Cor.~\ref{cor:const}, $\alpha$ has to be constant.
		In Cor.~\ref{cor:better-tim}, $\alpha(\cdot)$ is a function of~$\mu$.
		The ratio of \IM 2 is by Thm.~\ref{thm:2IM}; combining this with Thm.~\ref{thm:MIM-to-MUM} yields the ratio for~\UM 2.
		}
	\end{figure}

	\paragraph{Preprocessing and Observations.}
	Given a graph~$G = (V,E)$,
	a single edge~$e$ is \emph{allowed} if there exists a perfect matching~$M$ in~$G$ with $e \in M$ and \emph{forbidden} otherwise.
	A graph is \emph{matching-covered} if all its edges are allowed
		(cf.~\cite{LP86} for a concise characterization of matching-covered graphs).
	Forbidden edges can easily be found in polynomial time; see e.g.~\cite{RV89} for an efficient algorithm.
	A simple preprocessing for \MIM and \MUM is to remove the forbidden edges in each stage,
		as they will never be part of a multistage matching.
	Thereby, we obtain an equivalent \emph{reduced} temporal graph,
		i.e., a temporal graph whose stages are matching-covered.
	If any stage in the reduced temporal graph contains no edges (but vertices), the instance is infeasible.
	In the following, we thus assume w.l.o.g.\ that the given temporal graph is reduced and \emph{feasible}, i.e., in each stage there exists some perfect matching.

	\begin{observation}\label{observation: Trivial solution}
		Let \G be a reduced 2-stage graph. For any $e\in E_\cap$, there is a perfect matching in each stage that includes~$e$. Thus, there is a multistage perfect matching with profit at least $1$ if $E_\cap\neq\varnothing$.
	\end{observation}

	\begin{observation}\label{observation: Trivial 2-apx for MUM}
		For any multistage perfect matching $(M_i)_{i\in\range{\tau}}$, it holds for each $i\in\range{\tau-1}$ that $\max(n_i/2,n_{i+1}/2) \leq |M_i\cup M_{i+1}| = \uprofit(M_i,M_{i+1}) \leq 2\max(n_i/2,n_{i+1}/2)$.
		Thus, computing any multistage perfect matching is an immediate $2$-approximation for \MUM.
	\end{observation}

	\begin{observation}\label{observation: FPT-results}
		Consider the following algorithm:
		Enumerate every possible sequence $(F_i)_{i\in\range{\tau-1}}$ such that $F_i\subseteq E_i\cap E_{i+1}$ for each~$i\in\range{\tau-1}$;
		then check for each $i\in\range{\tau}$
		whether there is a perfect matching $M_i$ in $G_i$
		such that $F_{i-1}\cup F_i\subseteq M_i$, where $F_0 = F_\tau = \varnothing$.
		Thus,
		\MIM and \MUM are in FPT w.r.t.\ parameter $\sum_{i\in\range{\tau-1}} |E_i\cap E_{i+1}|$ (or similarly $\tau\cdot\mu$).
	\end{observation}

	\section{Setting the Ground}\label{sec:compl}
	Before we present our main contribution, the approximation algorithms,
		we motivate the intrinsic complexities of the considered problems.
	On the one hand, we show that the problem is already hard in very restricted cases.
	On the other hand, we show that natural linear programming methods cannot yield a constant-factor approximation for~\IM 2.

	While it is known that \IM 2 is NP-hard in general,
		we show that \IM 2 is already NP-hard
		in the seemingly simple case
		where each vertex has only degree~$2$ in both stages.
	It immediately follows that the decision variants of \MIM, \UM 2, \MUM, \MinMPM, and \MaxMPM remain NP-hard as well, even if restricted to this set of temporal graphs.
	See Appendix~\ref{appendix:hardness} for the proof of the following theorem:

	\newcommand{\thmhardness}{
		Deciding \IM 2 is NP-hard on spanning temporal graphs whose union graph is bipartite, even if both stages
		consist only of disjoint even cycles and $E_\cap$ is a collection of disjoint 2-paths.
	}
	\begin{theorem}\label{thm: 2SPM is NP-hard}
		\thmhardness
	\end{theorem}

	Linear Programs (LPs)---as relaxations of integer linear programs (ILPs)---are often used to provide dual bounds in the approximation context.
	Here, we consider the natural LP-formulation of \IM 2
	and show that the integrality gap (i.e., the ratio between the optimal objective value of the ILP and the optimal objective value of its relaxation) is at least~$\sqrt{\mu}$,
	even already for spanning instances with a bipartite union graph.
	Up to a small constant factor, this equals the (inverse) approximation ratio guaranteed by Algorithm~\ref{algo:Twostage approx}, which we will propose in \cref{sec:apx}.
	This serves as a hint that overcoming the approximation dependency $\sqrt{\mu}$ for \IM 2 may be hard.
	A proof of the following theorem can be found in Appendix~\ref{appendix:ILP}.

	\newcommand{\thmlpgap}{
		The natural LP for \IM 2 has at least an integrality gap of~$\sqrt{\mu}$,
		independent of the number $\mu$ of edges in the intersection.
	}
	\begin{theorem}\label{thm:LP-gap}
		\thmlpgap
	\end{theorem}

	\section{Approximation}\label{sec:apx}
	We start with the special case of \IM2, before extending the result to the multistage \MIM scenario.
	Then we will transform the algorithms for use with \UM2 and \MUM.

	\subsection{Approximating \IM2}\label{sec:2im-apx}
	We first describe Algorithm~\ref{algo:Twostage approx}, which is an approximation for \IM 2.
	Although its ratio is not constant but grows with the rate of $\sqrt{\mu}$,
		\cref{thm:LP-gap} hints that better approximations may be hard to obtain.
	Algorithm~\ref{algo:Twostage approx} roughly works as follows:
	Given a 2-stage graph $\G$, we iterate the following procedure on $G_1$ until every edge of $E_\cap$ has been in at least one perfect matching:
	Compute a perfect matching $M_1$ in $G_1$ that uses the maximum number of edges of $E_\cap$ that have not been used in any previous iteration;
	then compute a perfect matching $M_2$ in $G_2$ that optimizes the profit with respect to $M_1$.
	While doing so, keep track of the maximal occurring profit.
	Note that by choosing weights appropriately, we can construct a perfect matching that contains the maximum number of edges of some prescribed edge set in polynomial time~\cite{LP86}.

	\begin{algorithm}[tb]
		\caption{Approximation of \IM 2 \label{algo:Twostage approx}}
		set $(M_1, M_2)$ to $(\varnothing,\varnothing)$\;
		\For{$i = 1,2,...$}{
			set edge weights of $G_1$ to $\1{e\in E_\cap\setminus\bigcup_{j\in\range{i-1}}M_1^{(j)}}$\;
			compute a maximum weight perfect matching $M_1^{(i)}$ on $G_1$\;
			set edge weights of $G_2$ to $\1{e\in M_1^{(i)}}$\;
			compute a maximum weight perfect matching $M_2^{(i)}$ on $G_2$\;
			\lIf{$|M_1^{(i)}\cap M_2^{(i)}| \geq |M_1\cap M_2|$}{%
				set $(M_1, M_2)$ to $(M_1^{(i)}, M_2^{(i)})$}
			\lIf{$E_\cap \subseteq \bigcup_{j\in\range{i}} M_1^{(j)}$}{%
				\Return{$(M_1,M_2)$}}
		}
	\end{algorithm}

	We show:
	\begin{theorem}\label{thm:2IM}
		Algorithm~\ref{algo:Twostage approx} is a tight $(1/\!\sqrt{2\mu})$-approximation for \IM 2.
	\end{theorem}

	\noindent We prove this via two Lemmata;
	the bad instance of Lemma~\ref{lemma:2IM-tightness} in conjunction with the approximation guarantee (Lemma~\ref{lemma:2IM-guarantee}) establishes tightness.

	\newcommand{\lemmatightness}{The approximation ratio of Algorithm~\ref{algo:Twostage approx} is at most~$(1/\!\sqrt{2\mu})$.}
	\begin{lemma}[Bad instance]\label{lemma:2IM-tightness}
		\lemmatightness\ (Proof in Appendix~\ref{appendix:tightness}.)
	\end{lemma}

	\begin{lemma}[Guarantee]\label{lemma:2IM-guarantee}
		The approximation ratio of Algorithm~\ref{algo:Twostage approx} is at least~$(1/\!\sqrt{2\mu})$.
	\end{lemma}

	\begin{proof}\let\qed\relax
		Let $\G$ be a feasible and reduced 2-stage graph with non-empty $E_\cap$.
		Clearly, our algorithm achieves $\apx\geq 1$ as described in Observation~\ref{observation: Trivial solution}.
		Let~$k$ denote the number of iterations.
		For any~$i\in\range{k}$, let $(M_1^{(i)},M_2^{(i)})$ denote the 2-stage perfect matching computed in the $i$th iteration.
		The algorithm picks at least one new edge of~$E_\cap$ per iteration into $M_1^{(i)}$ and hence terminates.
		Let~$(M^*_1,M^*_2)$ denote an optimal 2-stage perfect matching and $M^*_\cap\coloneqq M^*_1\cap M^*_2$ its intersection.
		Let~$R_i \coloneqq (M_1^{(i)} \cap E_\cap)\setminus \bigcup_{j\in\range{i-1}}R_j$
			denote the set of edges in $M_1^{(i)}\cap E_\cap$ that are not contained in~$M_1^{(j)}$ for any previous iteration~$j < i$
			and let~$r_i \coloneqq |R_i|$.
		Note that in iteration~$i$, the algorithm first searches for a perfect matching~$M_1^{(i)}$ in $G_1$ that maximizes the cardinality~$r_i$ of its intersection with~$E_\cap\setminus\bigcup_{j\in\range{i-1}}R_j$.
		We define
			$R^*_i \coloneqq (M_1^{(i)} \cap M^*_\cap)\setminus \bigcup_{j\in\range{i-1}}R^*_j$
			and $r^*_i\coloneqq|R^*_i|$
			equivalently to~$R_i$, but w.r.t.~$M^*_\cap$
			(cf.~\cref{fig:m-and-r}).
		Observe that $R_i\cap M^*_\cap = R^*_i$.
		\begin{figure}[t]
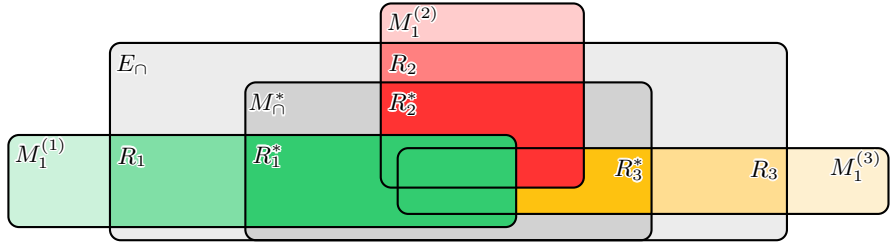

			\centering
			\include{tikz-inclusions}%
			\caption{
				Visualization of the relationships between $E_\cap, M^*_\cap, M^{(i)}_1, R_i$ and $R^*_i$ for $i\in\range{3}$.\label{fig:m-and-r}
			}
		\end{figure}

		Let~$q\coloneqq \sqrt{2\mu}$.
		For every $i\in\range{k}$ the algorithm chooses $M_2^{(i)}$ such that $|M_1^{(i)}\cap M_2^{(i)}|$ is maximized.
		Since we may choose $M_2^{(i)} = M^*_2$, it follows that $\apx \geq \max_{i\in\range{k}} r_i^*$.
		Thus, if $\max_{i\in\range{k}} r_i^* \geq \opt/q$, we have a $(1/q)$-approximation.
		In case $\opt\leq q$, any solution with profit at least~$1$ yields a $(1/q)$-approximation.
		We show that we are in one of these cases.

		Let $\qceil\coloneqq\lceil q\rceil$.
		Assume that $\opt > q$ (thus $\opt\geq \qceil$) and simultaneously $r_i^* < \opt/q$ for all $i\in\range{k}$.
		Since we distribute~$M^*_\cap$ over the disjoint sets $\{R^*_i\mid i\in\range{k}\}$, each containing less than $\opt/q$ edges, we know that $k > q$ (thus $k\geq\qceil$).
		In iteration $i$, $M^*_1$ has weight
			$|(M^*_1 \cap E_\cap) \setminus \bigcup_{j\in\range{i-1}}R_j|
			\geq |M^*_\cap \setminus \bigcup_{j\in\range{i-1}}R_j|
			= |M^*_\cap \setminus \bigcup_{j\in\range{i-1}}R^*_j|$.
		Hence, the latter term is a lower bound on~$r_i$, that we estimate as follows:
		$r_i\geq  \big|M^*_\cap\setminus \bigcup_{j\in\range{i-1}}R^*_j\big|
			= \opt -\sum_{j\in\range{i-1}}r^*_j \geq \opt - \sum_{j\in\range{i-1}} \opt/q
			= \opt \cdot\big(1 - (i - 1)/q\big)$.
		The above assumptions give a contradiction:%
		\begin{linenomath}
			{\setlength{\mathindent}{0cm}
			\begin{align*}
				\mu= \big|\bigcup_{i\in\range{k}} R_i\big|
					\geq
					\sum_{i\in\range{\qceil}} r_i
					\geq \opt \cdot \sum_{i\in\range{\qceil}} \big(1 - \frac{i - 1}{q}\big)
					\geq \qceil \cdot \sum_{i\in\range{\qceil}} \big(1 - \frac{i - 1}{q}\big)\\
					= \qceil \big(\qceil - \sum_{i\in\range{\qceil -1}}\frac iq\big)
					= \qceil \big(\qceil - \frac{(\qceil -1)\qceil}{2q}\big)
					= \qceil^2 \big(1 - \frac{\qceil -1}{2q}\big)
					\boldsymbol> \qceil^2 \big(1 - \frac{q}{2q}\big)
					\geq \mu.~\hfill\qedsymbol
			\end{align*}
			}
		\end{linenomath}
	\end{proof}

	\subsection{Approximating \MIM}\label{sec:mim-apx}
	Let us extend the above result to an arbitrary number of stages.
	We show that we can use \emph{any} \IM 2 approximation algorithm (in particular also Algorithm~\ref{algo:Twostage approx})
		as a black box to obtain an approximation algorithm for \MIM,
		while only halving the approximation ratio:
	Algorithm~\ref{algo:sequence-apx} uses an edge weighted path $(P,w)$ on $\tau$ vertices as an auxiliary graph.
	We set the weight of the edge between the $i$th and $(i+1)$th vertex
		to an approximate solution for the \IM 2 instance
		that arises from the $i$th and $(i+1)$th stage of the \MIM instance.
	A maximum weight matching $M_P$ in~$(P,w)$ induces a feasible solution for the \MIM problem:
	If an edge $(j,j+1)$ is in $M_P$,
		we use the corresponding solutions for the $j$th and $(j+1)$th stage;
		for stages without incident edge in~$M_P$,
		we select an arbitrary solution.
	Since no vertex is incident to more than one edge in~$M_P$,
		there are no conflicts.

	\begin{observation}\label{lem: Path matching}
		For~$F \subseteq E(P)$, denote $w(F) \coloneqq \sum_{e \in F} w(e)$.
		Let $e_i$ denote the $i$th edge of $P$.
		For $b\in\range{2}$, the matchings $M_b \coloneqq \{e_i\in E(P)\mid i=b\mod 2\}$
		are disjoint and their union is exactly $E(P)$.
		Thus, any maximum weight matching $M_P$ in $P$ achieves $2\cdot w(M_P)\geq w(E(P))$.
	\end{observation}

	\begin{algorithm}[t]
		\caption{General multistage approximation \label{algo:sequence-apx}}

		\KwIn{%
			Temporal graph~\G,
			2-stage perfect matching algorithm~$\mathcal A$
		}
		\BlankLine
		create path $P\coloneqq \{e_1,...,e_{\tau-1}\}$\;
		\ForEach{$i\in\range{\tau-1}$}{
			set $(S_i,T_{i+1})$ to $\mathcal A(V, E_i, E_{i+1})$
			\Comment*{approximate 2-stage graphs}
			set weight of $e_i$ to $w_i\coloneqq |S_i\cap T_{i+1}|$%
		}
		compute maximum weight matching $M_P$ in $(P,w)$\;
		set $(M_i)_{i\in\range{\tau-1}}$ to $(S_i)_{i\in\range{\tau-1}}$ and $M_\tau$ to $T_\tau$
		\Comment*{set initial solution}
		\ForEach(\tcp*[f]{modify solution according to $M_P$}){$i\in\range{\tau-1}$}{
			\lIf{$e_{i}\in M_P$}{%
				set $M_i$ to $S_i$ and $M_{i+1}$ to $T_{i+1}$}%
		}
		\Return{$(M_1,...,M_\tau)$}
	\end{algorithm}

	\begin{theorem}\label{thm:2IM-to-MIM}
		For a \IM 2 $\alpha$-approximation, Alg.\,\ref{algo:sequence-apx} $\frac\alpha2$-approximates \MIM.
	\end{theorem}

	\begin{proof}
		Let $\G = (V, E_1,..., E_\tau)$ be the given temporal graph.
		For any $i\in\range{\tau-1}$, $(S_i,T_{i+1})$ is the output of the \IM2 $\alpha$-approximation $\mathcal A(V, E_i, E_{i+1})$; let~$w_i\coloneqq|S_i\cap T_{i+1}|$.
		Let $\M^*\coloneqq(M_1^*,...,M_\tau^*)$ denote a multistage perfect matching whose profit $\profit(\M^*)$ is maximum.
		Since $\mathcal A$ is an $\alpha$-approximation for \IM 2, we know that $|M_i^*\cap M_{i+1}^*|\leq w_i/\alpha$ for every $i\in\range{\tau -1}$.
		Thus~$\profit(\M^*)\leq (1/\alpha) \sum_{i\in\range{\tau -1}} w_i$.
		Algorithm~\ref{algo:sequence-apx} computes a maximum weight matching $M_P$ in~$(P,w)$ and constructs a multistage solution~$\M$.
		By Observation~\ref{lem: Path matching}, we obtain $\profit(\M^*)\leq (1/\alpha)\sum_{i\in\range{\tau -1}} w_i = (1/\alpha)\cdot  w(E(P)) \leq (2/\alpha)\cdot w(M_P) \leq (2/\alpha)\cdot \profit(\M)$.
	\end{proof}

		We compute a maximum weight matching in a path in linear time using straightforward dynamic programming.
		Hence, assuming running time~$f$ for~$\mathcal A$, Algorithm~\ref{algo:sequence-apx} requires
		$\mathcal O\big(\sum_{i\in \range{\tau-1}}|f(G_i,G_{i+1})|\big)$ steps.

	\begin{corollary}\label{cor:MIM-approx}
		Alg.\,\ref{algo:Twostage approx} in Alg.\,\ref{algo:sequence-apx} yields a $(1/\!\sqrt{8\mu})$-approximation for \MIM.
	\end{corollary}

	There is another way to approximate \MIM via an approximation for \IM 2, which neither dominates nor is dominated by the above method:

	\begin{theorem}\label{thm:1-2-1}
		There is an S-reduction from \MIM to \IM 2, i.e.,
		given any \MIM instance~$\G$, we can find a corresponding \IM 2 instance~$\G'$ in polynomial time
		such that any solution for~$\G$ bijectively corresponds to a solution for~$\G'$ with the same profit.
		Furthermore, $|E(G'_1) \cap E(G'_2)| = \sum_{i \in \range{\tau-1}} |E(G_i) \cap E(G_{i+1})|$.
	\end{theorem}
	\begin{proof}
		We will construct a 2-stage graph~$\G'$ whose first stage~$G'_1$ consists of (subdivided) disjoint copies
		of $G_i$ for odd~$i$;
			conversely its second stage~$G'_2$ consists of (subdivided) disjoint copies of $G_i$ for even~$i$.
		More precisely, consider the following construction:
		Let~$b(i) \coloneqq 2 - (i \bmod 2)$.
		For each $i \in \range{\tau}$, we create a copy of $G_i$ in $G'_{b(i)}$
			where each edge~$e \in E(G_i)$ is replaced by a $7$-path~$p_i^e$.
		We label the $3$rd ($5$th) edge along~$p_i^e$ (disregarding its orientation) with $e_i^-$ ($e_i^+$, respectively).
		To finally obtain $\G'$, for each $i \in \range{\tau-1}$ and $e \in E(G_i) \cap E(G_{i+1})$,
			we identify the vertices of $e_i^+$ with those of $e_{i+1}^-$ (disregarding the edges' orientations);
			thereby precisely the edges $e_i^+$ and $e_{i+1}^-$ become an edge common to both stages.
		No other edges are shared between both stages.
		This completes the construction of~$\G'$
			and we have $|E(G'_1) \cap E(G'_2)| = \sum_{i \in \range{\tau-1}} |E(G_i) \cap E(G_{i+1})|$.

		Assume~$\M' \coloneqq (M'_1, M'_2)$ is a solution for~$\G'$.
		Clearly, each path~$p_i^e$ in $G'_{b(i)}$ is matched alternatingly and hence
			either all or none of $e_i^-,e_i^+$, the first, and the last edge of $p_i^e$ are in~$M'_{b(i)}$.
		We derive a corresponding solution~$\M$ for $\G$:
		For every $i \in \range{\tau}$ and $e \in E(G_i)$,
			we add $e$ to $M_i$ if and only if $e_i^- \in M'_{b(i)}$.
		Suppose that $M_i$ is not a perfect matching for $G_i$,
			i.e., there exists a vertex~$v$ in $G_i$
			that is not incident to exactly one edge in~$M_i$.
		Then also the copy of~$v$ in the copy of $G_i$ in $G'_{b(i)}$ is not incident to
			exactly one edge of $M'_{b(i)}$, contradicting the feasibility of~$\M'$.

		Consider the profit achieved by $\M$:
		Every edge in $M'_1\cap M'_2$ corresponds to a different identification $\langle e_i^+, e_{i+1}^-\rangle$.
		We have $e \in M_i \cap M_{i+1}$ if and only if $e_{i}^- \in M'_{b(i)}$, $e_{i+1}^- \in M'_{b(i+1)}$, and $e_i^+ = e_{i+1}^-$.
		It follows that this holds if and only if $e_{i}^+ \in M'_{b(i)} \cap M'_{b(i+1)}$
			and hence the profit of~$\M$ is equal to that of~$\M'$.
		The inverse direction proceeds in the same manner.
	\end{proof}

	Since the new $\mu':=|E(G'_1)\cap E(G'_2)|$ is largest w.r.t.\ the original~$\mu$ if $|E(G_i)\cap E(G_{i+1})|$ is constant for all $i$, we obtain:
	\begin{corollary}\label{cor:better-tim}
		For any \IM 2 $\alpha(\mu)$-approximation where $\alpha(\mu)$ is a (typically decreasing) function of~$\mu$,
		there is an $\alpha\big((\tau-1)\mu\big)$-approximation for \MIM.
		Using Algorithm~\ref{algo:Twostage approx},
			this yields a ratio of $1/\!\sqrt{2(\tau-1)\mu}$;
		for \IM 3 and \IM 4 this is tighter than~\cref{thm:2IM-to-MIM}.
	\end{corollary}

	Assume the approximation ratio for \IM 2 would not depend on~$\mu$.
	Then the above would yield a surprisingly strong result:
	\begin{corollary}\label{cor:const}
		Any \IM 2 $\alpha$-approximation with constant $\alpha$ results in an $\alpha$-approximation of \MIM.
		If \MIM is APX-hard, so is \IM 2.
	\end{corollary}

	\subsection{Approximating \MUM}\label{sec:mum-apx}
	Consider the~\MUM-problem which minimizes the cost.
	As noted in Observation~\ref{observation: Trivial 2-apx for MUM}, a $2$-approximation is easily accomplished.
	However, by exploiting the previous results for \MIM, we obtain better approximations.

	\begin{theorem}\label{thm:MIM-to-MUM}
		Any $\alpha$-approximation of \MIM is a $(2-\alpha)$-approximation of \MUM.
	\end{theorem}
	\begin{proof}
		Recall that an optimal solution of \MIM constitutes an optimal solution of \MUM.
		As before, we denote the heuristic sequence of matchings by~$(M_i)_{i\in\range{\tau}}$ and the optimal one by $(M^*_i)_{i\in\range{\tau}}$.
		Let $\xi \coloneqq  \sum_{i \in \range{\tau-1}} (n_i + n_{i+1}) / 2$.
		Consider the solutions' values w.r.t.\ \MUM:
		\begin{linenomath}
		\[
			\frac{\apx_\cup}{\opt_\cup}
			=    \frac{\sum_{i \in \range{\tau-1}}\uprofit(M_i,M_{i+1})}{\sum_{i \in \range{\tau-1}}\uprofit(M^*_i,M^*_{i+1})}
			=    \frac{\xi - \sum_{i \in \range{\tau-1}}|M_i\cap M_{i+1}|}{\xi - \sum_{i \in \range{\tau-1}}|M^*_i\cap M^*_{i+1}|}
			\leq \frac{\xi - \alpha\cdot\opt_\cap}{\xi - \opt_\cap}
			=: f.
		\]
		\end{linenomath}
		As $0 < \alpha < 1$, $f$ is monotonously increasing in $\opt_\cap$ if $0 \leq \opt_\cap < \xi$.
		Thus, since $\opt_\cap \leq \sum_{i \in \range{\tau-1}} \min(n_i,n_{i+1})/2 \leq \sum_{i \in \range{\tau-1}} (n_i+n_{i+1})/4 = \xi/2$,
			it follows that $\apx_\cup / \opt_\cup \leq (\xi-\alpha\cdot\xi/2) / (\xi - \xi/2) = 2-\alpha$.
	\end{proof}

	\begin{corollary}\label{cor:UM-approx}
		Let $r \coloneqq \min\{8,2(\tau-1)\}$.
		We have a $\big(2-1/\!\sqrt{r\cdot\mu}\big)$-approximation for \MUM.
	\end{corollary}

	Note that a similar reduction from \MIM to \MUM is not achieved as easily:
	Consider any $(1+\varepsilon)$-approximation for \MUM.
	Choose an even integer~$k\geq 6$ such that $k/(k-1) \leq 1+\varepsilon$; consider a spanning $2$-stage instance
		where each stage is a $k$-cycle and $E_\cap$ consists of a single edge~$e$.
	The optimal 2-stage perfect matching~$\M^*$ that contains~$e$ in both stages has profit~$\profit(\M^*) = 1$ and cost~$\uprofit(\M^*) = 2\cdot k/2 - 1 = k - 1$.
	A 2-stage perfect matching $\M$ that does not contain~$e$ still satisfies $\uprofit(\M) = k$ and as such is an $(1+\varepsilon)$-approximation for \MUM.
	However, its profit~$p(\M) = 0$ does not provide any approximation of $\profit(\M^*) = 1$.

	As for \MIM, we aim to extend a given approximation for \UM 2 to a general approximation for \MUM.
	Unfortunately, we cannot use~\cref{thm:1-2-1,thm:MIM-to-MUM} for this, as an approximation for \UM 2 does not generally constitute one for \IM 2 (and \MIM).
	On the positive side, a similar approach as used in the proof of~\cref{thm:2IM-to-MIM} also works for minimization.

	\begin{theorem}\label{thm:2UM-to-MUM}
	Any $\alpha$-approximation $\mathcal A$ for \UM 2 results in a $(1 + \alpha/2)$-approximation for \MUM by using $\mathcal A$ in  Algorithm~\ref{algo:sequence-apx}.
	\end{theorem}

	\begin{proof}
		As before, let $(M_i^*)_{i\in\range{\tau}}$ denote an optimal solution for \MUM.
		For each $i\in\range{\tau -1}$, $(S_i,T_i)$ denotes the output of $\mathcal A(V, E_i, E_{i+1})$.
		For $L\subseteq \range{\tau-1}$, let $\xi(L) \coloneqq \sum_{i\in L}(n_i + n_{i+1})/2$ and $\sigma(L)\coloneqq \sum_{i\in L}|S_i\cup T_i|$.
		Note that $w_i\coloneqq\xi(\{i\}) - \sigma(\{i\})$ equals the weight of $e_i$.
		We define $I\coloneqq \{i\in\range{\tau-1} \mid e_i\in M_P\}$ as the set of indices corresponding to $M_P$ and $J\coloneqq \range{\tau-1}\setminus I$ as its complement.
		By Observation~\ref{lem: Path matching}, we have $w\big(E(P)\big) \leq 2\cdot w(M_P)$, thus
		\begin{linenomath}
		\begin{align*}
		\xi(I) - \sigma(I) + \xi(J) -\sigma(J) = w\big(E(P)\big) \leq 2\cdot w(M_P) = 2\big(\xi(I) - \sigma(I)\big)\\
		\mathllap{\Rightarrow}\,\sigma(I) + \xi(J) \leq \xi(I) + \sigma(J)
			\Rightarrow 2\big(\sigma(I) + \xi(J)\big)
			\leq \xi(I\cup J) + \sigma(I\cup J).
		\end{align*}
		\end{linenomath}
		\noindent The trivial upper bound $\xi$ suffices to bound the algorithm's solution value:%
		\begin{linenomath}
		\[
		\apx
		= \sigma(I) + \sum_{j\in J}|M_j\cup M_{j+1}|
		\leq \sigma(I) + \xi(J)
		\leq\frac12\big(\xi(I \cup J) + \sigma(I \cup J)\big).\\
		\]
		\end{linenomath}
		Since $\sigma(I\cup J)$ $\alpha$-approximates the sum of all \UM 2 instances' solution values, we have $ \sigma(I \cup J) \leq \alpha\cdot\opt$.
		For each transition, any solution satisfies $(n_i + n_{i+1})/4 \leq |M_i \cup M_{i+1}|$ and hence $\xi(I\cup J) \leq 2\cdot \opt$.
		Finally, we obtain the claimed ratio: $\apx \leq 1/2\cdot\big(2\cdot \opt + \alpha\cdot \opt\big) = (1+\alpha/2)\cdot\opt$.\qedhere
	\end{proof}

\section{Conclusion}
In this paper we presented the first approximation algorithm for \IM 2, having a tight approximation ratio of $1/\!\sqrt{2\mu}$.
It remains open if a constant factor approximation for \IM 2 is possible;
	however, we showed that this would imply a constant factor approximation for \MIM.
We further showed two ways in which \MIM and \MUM can be approximated by using any algorithm that approximates \IM 2,
	thereby also presenting the first approximation algorithms for multistage matching problems with an arbitrary number of stages.
We are confident that our techniques are applicable to a broader set of related problems as well.

\bibliography{main}

\clearpage

\appendix
\section*{APPENDIX}

\section{A note on approximating \MaxMPM}\label{appendix:MaxMPM}

The proposed $(1/2)$-approximation for \MaxMPM in~\cite{BELP18} does not work.
It takes a temporal graph as input, where each stage may be an arbitrary graph (not necessarily complete),
picks a matching for every second stage~$G_i$, and reuses the same matching for stage $G_{i+1}$.
Thus, every second stage transition is optimal, whereas every other second transition potentially constitutes a worst case. \emph{If} the algorithm's solution is feasible, we indeed yield the proposed approximation ratio.
However, such an approach is inherently problematic as there is no reason why a matching in $G_i$ would need to be feasible for $G_{i+1}$.
In fact, consider a temporal graph $\G=(V,E_1,...,E_\tau)$ with $V=\{v_1,...,v_4\}$.
Let $E_i = \{v_1v_2,v_3v_4\}$ for odd~$i$,
and $E_i = \{v_2v_3,v_4v_1\}$ for even~$i$.
No perfect matching in $E_i$ is also a perfect matching in~$E_{i+1}$.

Thus, although any $\alpha$-approximation for \MaxMPM would directly yield an $\alpha$-approximation for \MIM on spanning temporal graphs, we currently do not know of any such algorithm.
In fact, a constant-factor approximation seems difficult to obtain, cf.~\cref{thm:LP-gap}.
Personal communication with B.~Escoffier confirmed our counterexample.
One may consider a relaxed version of \MaxMPM where one tries to find matchings of large weight in each stage, formally optimizing the weighted sum between the profit and the summed stagewise matching weights.
Observe that in this scenario it is not guaranteed that the optimal solution induces a perfect (nor even maximum) matching in each stage.
However, for this problem their analysis would be correct and their algorithm yields a $(1/2)$-approximation.

\paragraph{Counterexample.}

We examine the $(1/2)$-approximation algorithm~$\mathcal A$ for \MaxMPM that was proposed in~\cite[Theorem~8]{BELP18},
and give a reduced spanning $4$-stage instance
where~$\mathcal A$ does not yield a feasible solution.
We use four stages since the algorithm treats fewer stages as special cases.
Still, the feasibility problem that we are about to describe is inherent to all its variants.

Consider the temporal graph~$\G=(V,E_1,E_2,E_3,E_4)$ given in~\cref{fig:counterexample-algo}, where $E_1 = E_3$ and $E_2 = E_4$ and edges have uniform weight~$0$. We trivially observe that any perfect matching in~$G_i$ is optimal w.r.t.\ edge weight.
For each~$i\in\range{3}$, we have $E_\cap^i \coloneqq  E_i\cap E_{i+1} = \{e_1,e_2,e_3,e_4\}$.
\begin{figure}[tb]
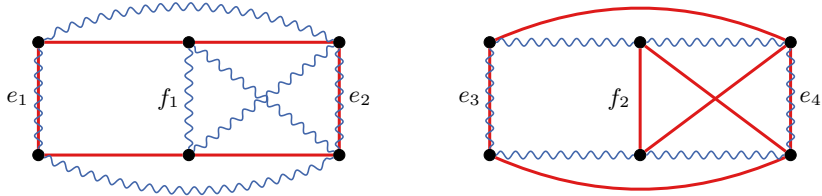

	\centering
	\include{tikz-belp}
	\caption{
		Counterexample for the proposed $(1/2)$-approximation. Edges in $E_1{=}E_3$ are curvy (and blue), edges in $E_2{=}E_4$ are straight (and red).
	}\label{fig:counterexample-algo}
\end{figure}

The algorithm proceeds as follows on~$\G$:
For each~$i\in\range{3}$, it computes a perfect matching~$M_i$ in $G_i$ that maximizes~$|M_i \cap E_{i+1}|$.
It constructs the solutions $\M \coloneqq (M_1,M_1,M_3,M_3)$ and $\M' \coloneqq (\hat{M}_1,M_2,M_2,\hat{M}_3)$, where~$\hat{M}_i$ is an arbitrary perfect matching in~$G_i$,
and outputs the solution that maximizes the profit.

Any perfect matching~$M_1$ in~$G_1$ that maximizes~$|M_1 \cap E_2|$
must contain both $e_1$ and $e_2$ and as such also~$f_1$.
This contradicts the feasibility of~$\M$, since $f_1\not\in E_2$.
Conversely, any such perfect matching~$M_2$ in $G_2$
must contain both  $e_3$ and $e_4$ and as such also $f_2$.
Again, this contradicts the feasibility of $\M'$, since $f_2\not\in E_3$.
It follows that the algorithm cannot pick a feasible solution.

We are not aware of any way to circumvent this problem.

\section{Proof of \cref{thm: 2SPM is NP-hard}}\label{appendix:hardness}

\setcounterref{theorem}{thm: 2SPM is NP-hard}
\addtocounter{theorem}{-1}

\begin{theorem}
	\thmhardness
\end{theorem}

\begin{proof}
	We will perform a reduction from \textsf{MaxCut}~\cite{GJ79} to \IM 2.
	In \textsf{MaxCut}, one is given an undirected graph $G=(V,E)$, a natural number $k$ and the question is to decide whether there is an $S\subseteq V$ such that $|\delta(S)|\geq k$.
	In the first stage, we will construct an even cycle for each vertex and each edge of the original graph
	and in the second stage we will create an even cycle for each incidence between an edge and a vertex (cf. \cref{fig:hardness}).
	A perfect matching in the first stage will correspond to a vertex selection
	and a perfect matching in the second stage will allow us to count the edges that are incident to exactly one selected~vertex.

	Given an instance $\mathcal I \coloneqq (G=(V,E),k)$ of \textsf{MaxCut}, we construct an instance $\mathcal J \coloneqq (\G, \kappa)$ of \IM 2.
	Set $\kappa \coloneqq 3|E| + k$.
	We start with an empty 2-stage graph $\G\coloneqq (V',E_1,E_2)$.

	Let $I\coloneqq\{(v,e) \mid v\in V, e\in \delta(v)\}$ be the set of incidences.
	For each $(v,e) \in I$, we add two new disjoint 2-paths to $E_1\cap E_2$ and call them $X^e_v$ and $Y^e_v$.
	Mark one edge of each $X^e_v$ as~$x^e_v$ and one edge of each $Y^e_v$ as $y^e_v$.
	We will refer to the endpoint of $X^e_v$~($Y^e_v$) incident to $x^e_v$~($y^e_v$) as the \emph{marked endpoint} of $X^e_v$~(respectively~$Y^e_v$).

	In $G_1$, for each $e=vw\in E$, we generate a $6$-cycle through $Y^e_v$ and~$Y^e_w$ by adding an edge between the marked endpoint of one path and the unmarked endpoint of the other, and vice versa.
	Furthermore, for each $v\in V$, we generate a cycle of length $4|\delta(v)|$ through the paths $X(v) \coloneqq \{X^e_v\mid e\in\delta(v)\}$ as follows:
	Denote the elements of $X(v)$ as an arbitrarily ordered sequence $(X_i)_{i\in\range{\delta(v)}}$.
	For each $i\in\range{\delta(v)}$, connect the marked endpoint of $X_i$ to the unmarked endpoint of $X_{(i\bmod \delta(v))+1}$ by a $2$-path, each with a new inner vertex.
	In $G_2$, for each $(v,e)\in I$, we generate a $6$-cycle through $X^e_v$ and $Y^e_v$ by adding an edge between the marked and an edge between the unmarked endpoints of the $2$-paths, respectively.
	$G_1$ consists of $|V| + |E|$ and $G_2$ of $2|E|$ disjoint even cycles, thus $\G$ is reduced.

	\begin{figure}[tb]
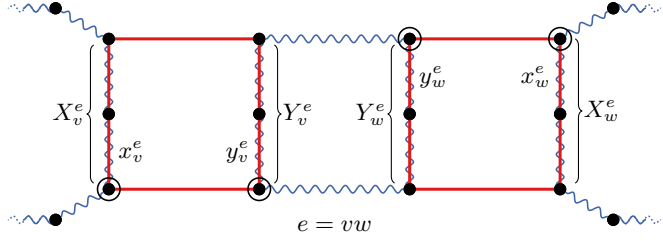

		\centering
		\include{tikz-hardness}
		\caption{
			Thm.~\ref{thm: 2SPM is NP-hard}: $E_1$ is curvy blue, $E_2$ is straight red.
			Marked vertices are encircled.
		}\label{fig:hardness}
	\end{figure}

	\paragraph{Claim.}
	\emph{$\mathcal J$ is a yes-instance if and only if $\mathcal I$ is a yes-instance.}

	\begin{proof}[Proof of Claim]\let\qed\relax
		Since both stages of $\G$ consist only of pairwise disjoint even cycles and there are only two perfect matchings in an even cycle, a perfect matching in a stage is determined by choosing one edge in each cycle. For $e=vw\in E$, let $X^e \coloneqq E(X^e_v)\cup E(X^e_w)$ and $Y^e \coloneqq E(Y^e_v)\cup E(Y^e_w)$.
		Observe that for any multistage perfect matching $(M_1,M_2)$, $M_1\cap M_2\subseteq E_\cap = \biguplus_{e\in E} X^e\cup Y^e$.

		\noindent "$\Leftarrow$" \tabto{10mm}
		Suppose there is an $S\subseteq V$, such that $|\delta(S)|\geq k$.
		For each $(v,e)\in I$, add $x^e_v$ to both $M_1$ and $M_2$ if $v\in S$.
		Otherwise, add the unmarked edge of $X^e_v$ to $M_1$ and $M_2$.
		This uniquely determines a perfect matching~$M_2$ in~$G_2$, where $y^e_v\in M_2 \iff x^e_v\in M_2 \iff v\in S$.
		For each $e=vw\in \delta(S)$ with $v\in S$ and $w\not\in S$, add $y^e_v$ to $M_1$; thus $y^e_w\not\in M_1$ and $y^e_v\in M_1\cap M_2$.
		For $e=vw\not\in \delta(S)$, add either $y^e_v$ or $y^e_w$ to $M_1$ (chosen arbitrarily).
		This determines a 2-stage perfect matching $(M_1,M_2)$.

		Consider some edge $e=vw\in E$.
		The intersection $M_1\cap M_2$ contains two edges of $X^e$.
		If $e\in \delta(S)$, it also contains two edges of $Y^e$, one marked and one unmarked.
		If $e\not\in \delta(S)$, $M_1\cap M_2$ contains exactly one edge of~$Y^e$.
		Thus, $|M_1\cap M_2| = |\biguplus_{e\in E} M_1\cap M_2\cap X^e| + |\biguplus_{e\in E} M_1\cap M_2\cap Y^e| = 3|E| + |\delta(S)|$.

		\newcommand{\qw}{m}
		\noindent "$\Rightarrow$" \tabto{10mm}
		Let $(M_1,M_2)$ be a multistage perfect matching in $\G$ with $|M_1\cap M_2|\geq 3|E| + k$.
		Observe that $|M_1\cap M_2| = \sum_{e\in E}\qw_e$ with $\qw_e \coloneqq|M_1\cap M_2\cap (X^e\cup Y^e)|\leq 4$ for each $e\in E$.
		Thus, by pigeonhole principle, there are at least $k$ edges with $\qw_e= 4$.

		$M_1$ yields a selection $S\subseteq V$:
		Select $v\in V$ if and only if $X(v)\subseteq M_1$.
		Observe that either all or none of the edges in $X(v)$ are matched simultaneously in a perfect matching in $G_1$.
		It can be seen that $\qw_e = 4$ if and only if $e\in \delta(S)$,
		thus $|\delta(S)|\geq k$.\hfill$\triangleleft$
	\end{proof}

	The cycles of length $4|\delta(v)|$ may have introduced an even number of vertices~$W\subseteq V$ that are isolated in $G_2$.
	To make \G spanning, we add to~$E_2$ an even cycle on~$W$.
	This neither interferes with $E_\cap$ nor the profit~$\profit$, since $W$ is an independent set in the first stage~$G_1$.
\end{proof}

\section{Proof of \cref{thm:LP-gap}}\label{appendix:ILP}

In the context of classical (perfect) matchings, the standard ILP formulation and its LP-relaxation describe the very same feasible points (called \emph{matching polytope}),
which is the corner stone of the problem being solvable in polynomial time~\cite{LP86}.
Given a 2-stage graph $\G=(V,E_1,E_2)$, the natural LP-formulation for \IM 2 starts with the product of two distinct such perfect matching polytopes.
Let $\delta_\ell(v)$ denote all edges incident to vertex~$v$ in~$G_\ell$, and
let $(M_1,M_2)$ be a 2-stage perfect matching in \G. For each $\ell\in\range{2}$, we model $M_\ell$ via the standard matching polytope:
For each~$e\in E_\ell$ there is an indicator variable $x^\ell_e$ that is $1$ if and only if~$e\in M_\ell$.
The constraints~\eqref{ILP:matching} below suffice for bipartite graphs; for general graphs one also
considers the \emph{blossom constraints}~\eqref{ILP:blossom}.
Additionally to these standard descriptions, for each~$e\in E_\cap$ we use a variable $z_e$ that is~$1$ if and only if $e\in M_1\cap M_2$.
We want to maximize~$\profit(M_1,M_2) = \sum_{e\in E_\cap} z_e$, such that:

\vspace{-\baselineskip}
\begin{subequations}
	\begin{align}
		& \sum_{e\in\delta_\ell(v)} x^\ell_e = 1		&& \forall \ell\in\range{2}, \forall v\in V(E_\ell)	\label{ILP:matching}\\
		& \sum_{e\in W} x^\ell_e \leq \frac{|W|-1}{2}	&& \forall \ell\in\range{2}, \forall W\subseteq V(E_\ell)\text{ with $|W|$ odd}	\label{ILP:blossom}\\
		& z_e \leq x^\ell_e								&& \forall \ell\in\range{2}, \forall e\in E_\cap	\label{ILP:min}\\
		& x^\ell_e \in \{0,1\}							&& \forall \ell\in\range{2}, \forall e\in E_\ell\\
		& z_e \in \{0,1\}								&& \forall e\in E_\cap
	\end{align}
\end{subequations}

Thereby, constraints~(\ref{ILP:min}), together with the fact that we maximize all $z$-values, ensure that $z_e = \min\{x^1_e,x^2_e\}$ in any optimal solution.

\begin{figure}
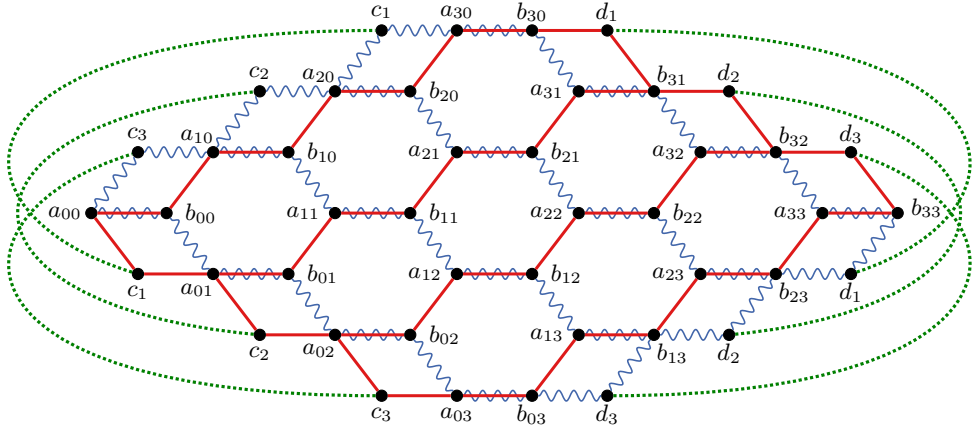

	\centering
	\include{tikz-lpgap}
	\caption{
		\IM 2 instance for $k{=}3$ with an integrality gap ${\geq}\sqrt{\mu}$.
		$E_1$ is straight and red, $E_2$ is curvy and blue.
		Dotted green lines identify vertices.
	}\label{fig:ILP gap}
\end{figure}

\setcounterref{theorem}{thm:LP-gap}
\addtocounter{theorem}{-1}

\begin{theorem}
	\thmlpgap
\end{theorem}

\begin{proof}
	We construct a family of \IM 2 instances, each with bipartite union graph, parameterized by some parameter $k$. Each instance is reduced and has a maximum profit of $1$, but its LP relaxation has objective value at least $k+1 = \sqrt{\mu}$.

	Fix some $k\geq 3$.
	We construct $\mathcal{G}\coloneqq \mathcal{G}(k) = (V,E_1,E_2)$ as follows (see \cref{fig:ILP gap} for a visualization with $k=3$).
	Let $V \coloneqq \{a_{i,j},b_{i,j}\}_{i,j\in\Range{k}} \cup \{c_i,d_i\}_{i\in\range{k}}$.
	Let $E_\cap \coloneqq \{ a_{i,j}b_{i,j} \}_{i,j\in\Range{k}}$, i.e., the intersection contains precisely the natural pairings of the $a$ and $b$ vertices.
	We call these the \emph{shared} edges.
	In $E_1$, we additionally add edges $\{b_{i-1,j}a_{i,j}\}_{i\in\range{k},j\in\Range{k}}$.
	Similarly, we add edges $\{b_{i,j-1}a_{i,j}\}_{i\in\Range{k},j\in\range{k}}$ to $E_2$.
	Now, both stages consist of $k+1$ disjoint paths of length $2k+1$ which are ``interwoven'' between the stages such that (i) every second edge in each path is shared (starting with the first), and (ii)
	any path in $G_1$ has exactly one edge in common with every path in~$G_2$.
	Let $P^\ell_i$, $\ell\in\range{2}, i\in\Range{k}$, denote those paths in their natural indexing.
	We make each stage connected by joining every pair of ``neighboring'' paths, each together with a $c$ and a $d$ vertex.
	More precisely, we add edges $\{c_ja_{0,j-1},c_ja_{0,j},b_{k,j-1}d_j,b_{k,j}d_j\}_{j\in\range{k}}$ to~$E_1$.
	Analogously, we add $\{c_{\varphi(i)}a_{i-1,0},c_{\varphi(i)}a_{i,0},b_{i-1,k}d_{\varphi(i)},b_{i,k}d_{\varphi(i)}\}_{i\in\range{k}}$ to $E_2$; the indexing function $\varphi(i)\coloneqq k-i+1$ ensures that these new edges are not common to both stages. (In fact, if we would not care for a spanning \G, we could simply use ``new'' vertices instead of reusing $c,d$ in $G_2$.)
	This finishes the construction, and since \G contains no forbidden edges, it is reduced.

	Since the inner vertices of any path $P^1_i$ have degree~2 in $G_1$, any perfect matching in $G_1$ either contains all or none of the path's shared edges.
	Assume some shared edge $a_{0,j}b_{0,j}$ is in a perfect matching in $G_1$.
	Let $C$ be the path $a_{0,0}\, c_1\, a_{0,1}\, c_2\,...\, c_k\, a_{0,k}$ in $G_1$.
	Recall that all $c$-vertices have degree~2 in $G_1$. Since $a_{0,j}$ is matched outside of $C$, all other $a_{0,j'}$, $j'\neq j$, have to be matched with these $c$-vertices.
	Thus, $P^1_j$ is the only path that contributes shared edges to the matching. Conversely, since $C$ contains one less $c$-vertex than $a$-vertices, any perfect matching in $G_1$ has to have at least (and thus exactly) one such path.
	As the analogous statement holds for $G_2$ and by the interweaving property (ii) above, any multistage perfect matching contains exactly one shared edge.

	However, we construct a feasible fractional solution with objective value~$\sqrt{\mu}$:
	\newcommand{\K}{\ensuremath{\lambda}}
	Let~$\K\coloneqq 1/(k+1)$.
	We set the $x$- and $z$-variables of all shared edges to $\K$, satisfying all constraints~\eqref{ILP:min}.
	This uniquely determines all other variable assignments, in order to satisfy~\eqref{ILP:matching}:
	Since the inner vertices of each $P^\ell_i$ have degree $2$ in $G_\ell$, the non-shared edges in these path have to be set to $1-\K$.
	Again consider path $C$: Each $a$-vertex in $C$ has an incident shared edge that contributes $\K$ to the sum in the vertex' constraint~\eqref{ILP:matching};
		there are no other edges incident to $C$.
	Thus, we have to set $x^1_{c_j a_{0,j-1}}=1-j\K$ and $x^1_{c_j a_{0,j}}=j\K$ such that, for each vertex in $C$, its incident variable values sum to 1.
	The analogous statements holds for the corresponding path through $d$-vertices in $G_1$, and the analogous paths in~$G_2$.
	All constraints~\eqref{ILP:matching} are satisfied.
	The blossom constraints~\eqref{ILP:blossom} act only on $x$-variables, i.e., on individual stages.
	Since our graph is bipartite, only considering the $x$-variables of one stage and disregarding \eqref{ILP:blossom} yields the bipartite matching polytope which has only integral vertices; our (sub)solution is an element of this polytope. Thus, \eqref{ILP:blossom} cannot be violated by our assignment.

	By construction we have $\mu=(k+1)^2$. Thus, the objective value of our assignment is $\sum_{e\in E_\cap} \K = \mu/(k+1) 
	= \sqrt{\mu}$, as desired.
\end{proof}

\section{Proof of Lemma~\ref{lemma:2IM-tightness}}\label{appendix:tightness}

\setcounterref{theorem}{lemma:2IM-tightness}
\addtocounter{theorem}{-1}

\begin{lemma}
	\lemmatightness
\end{lemma}

\newcommand{\marked}[1]{\ensuremath{\underline{#1}}}

\begin{proof}
Consider the following family~$\G_k$ of \IM2 instances, parameterized by number~$k\geq 1$.
An example using $k=4$ is depicted in \cref{fig:tightness}.
In the first stage, for each $i\in\range{k}$ create a $4$-cycle $C_i$ and label two of its adjacent vertices $w'_i$ and $w_i$.
Add a $3$-path with new inner vertices of degree~$2$ from~$w_{i}$ to $w'_{i+1}$ for each $i\in\range{k-1}$.
For each $i\in\range{k-1}$, create a vertex~$v_i$ and an edge $w_i v_i$.
Create a vertex $u$ and an edge $u w_k$.
For each $i\in\Range{k-1}$, create a vertex $u_i$ and an edge $u u_i$.
For each $i\in\range{k-1}$, create a path~$P_i$ from $u_i$ to $v_i$ with $2i+1$ edges
and label the new inner vertices with $a^i_1,b^i_1,a^i_2,b^i_2,...,a^i_i,b^i_i$ in this order.
\begin{figure}[tb]
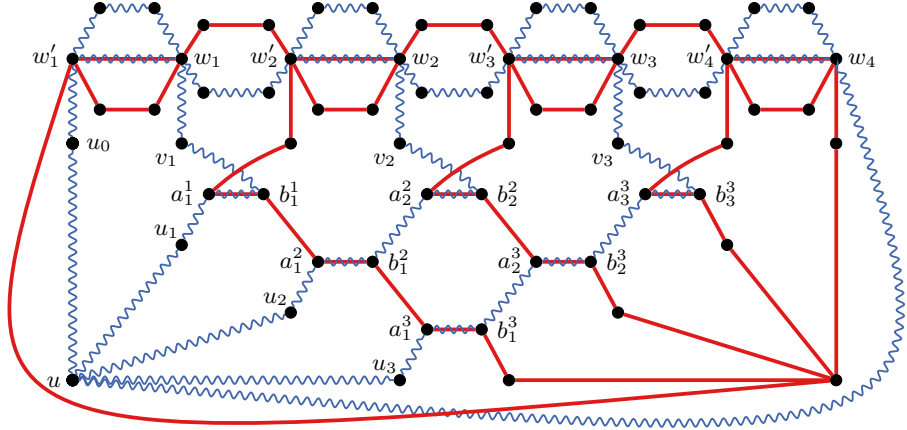

	\centering
	\include{tikz-tightness}
	\caption{
		\IM2 instance $\G_4$ as in Lemma~\ref{lemma:2IM-tightness}.
		Edges in $E_1$ are curvy and blue, edges in $E_2$ straight and red.
		The vertices are labeled according to the first stage.
	}\label{fig:tightness}
\end{figure}%

The second stage is constructed isomorphically to the first stage.
To avoid ambiguity in the naming, we underline element names of the second stage.
The $2$-stage graph $\G_k$ is completely defined by the following identifications:
	for each $i\in\range{k}$,
		let $\marked{w}'_i = w_{k-i+1}$
		and $\marked{w}_i = w'_{k-i+1}$;
	for each $i\in\range{k-1}$ and each $j\in\range{i}$,
		let $\marked{a}^i_j = b^{k-j}_{k-i}$
		and $\marked{b}^i_j = a^{k-j}_{k-i}$.
Thus, $E_\cap = F\cup A$ is precisely the union of $F\coloneqq\{w_i w'_i\mid i\in\range{k}\}$ and $A\coloneqq\{a^i_j b^i_j\mid i\in\range{k-1},j\in\range{i}\}$.
Observe that $\G_k$ is reduced, its union graph is bipartite,
	and $\mu = k + \sum_{i \in \range{k-1}} i = k(k+1)/2$.

By construction, for any perfect matching~$M$ in the first (second) stage, $|M\cap E_\cap|\leq k$.
Let~$M_F$ ($\marked{M}_F$) denote the unique perfect matching that
	contains $uu_0$ ($\marked{u}\hspace*{1pt}\marked{u}_0$, respectively) and all of $F$.
The pair $(M_F,\marked{M}_F)$ is an optimal solution with profit~$|F|=k$.

Consider an alternative perfect matching~$M$ in the first stage.
In each cycle $C_i$, we consider the shared edge $w_i w'_i$
	and its \emph{opposing edge} (i.e., its unique non-adjacent edge in $C_i$).
We distinguish between three possibilities regarding their memberships in $M$:
	the shared edge and its opposing edge are in~$M$ (\emph{type~Y}),
	only the opposing edge is in $M$ (\emph{type~N1}),
	none of them are in $M$ (\emph{type~N2}).

Picking an edge adjacent to $u$ determines a perfect matching up to the types of some $C_i$-cycles.
For $i\in\Range{k-1}$,
	let $M_i$ denote the unique perfect matching that contains $u u_i$,
	is type~\emph{Y} in $C_{i+1}$,
	but type~\emph{N2} in each $C_j$ for $j > i+1$.
Note that $M_i$
	is type~\emph{N1} in each $C_j$ with $j \leq i$
	and contains all $i$ shared edges along $P_i$.
Thus, $|M_i \cap E_\cap| = i+1$.
Most importantly,
	consider any perfect matching $M'$ in the second stage.
By construction,
	for any $i\in\Range{k-1}$,
	no two edges of $(M_i\cap A)\cup\{w'_{i+1} w_{i+1}\}$ can be contained simultaneously in $M'$.
It follows that $|M_i\cap M'| \leq 1$,

Algorithm~\ref{algo:Twostage approx} may never choose the optimal $M_F$ as a perfect matching for the first stage:
In the first iteration,
	both $M_F$ and $M_{k-1}$ have weight $k$,
	so the algorithm may choose $M_{k-1}$
	and obtain a $2$-stage perfect matching with profit~$1$.
In the following iteration,
	the weight (denoting the preference of edges) of $M_F$ is decreased by~$1$,
	since the edge $w'_k w_k$ has already been chosen in $M_{k-1}$.
Consequently, in each following iteration~$i\in\range{k}$
	the algorithm may choose $M_{k-i}$ over $M_F$,
	each time decreasing the weight of~$M_F$ by~$1$.
After choosing $M_0$ over $M_F$ in iteration~$k$,
	each edge in $E_\cap$ has been in some matching in the first stage;
	the algorithm stops and returns a $2$-stage perfect matching with profit~$1$.

Since $\mu = k(k+1)/2$, the optimal profit is $k = \big(\sqrt{8\mu+1}-1\big)/2$.
Thus, the approximation factor is at most $1/k = 2/(\sqrt{8\mu+1}-1\big)$ which tends to our guarantee of~$1/\!\sqrt{2\mu}$ for increasing~$k$.
\end{proof}
\end{document}

%% file: tikz-inclusions.tex
\begin{tikzpicture}[text height=12pt,xscale=.9,yscale=.7]
	\definecolor{mygreen}{rgb}{0, 0.75, 0.3}
	\newcommand{\yelloworange}{yellow!50!orange}
	\newcommand{\lightgray}{gray!15}
	\newcommand{\darkgray}{gray!35}

	\useasboundingbox[blue] (0,.75) rectangle (10,3.75);
	\tikzset{box/.style={rounded corners,thick}}

	\tikzset{firstBoxOuter/.style={box,mygreen!20}}
	\tikzset{firstBox/.style={box,mygreen!50}}
	\tikzset{firstBoxInner/.style={box,mygreen!80}}

	\tikzset{secondBoxOuter/.style={box,red!20}}
	\tikzset{secondBox/.style={box,red!50}}
	\tikzset{secondBoxInner/.style={box,red!80}}

	\tikzset{thirdBoxOuter/.style={box,\yelloworange!20}}
	\tikzset{thirdBox/.style={box,\yelloworange!50}}
	\tikzset{thirdBoxInner/.style={box,\yelloworange}}

	\newcommand{\eCapBox}{(0, 0) rectangle (10, 3.75)}
	\newcommand{\mCapBox}{(2, 0) rectangle (8, 3)}
	\newcommand{\firstMBox}{(-1.5,0.25) rectangle (6,2)}
	\newcommand{\secondMBox}{(4,1) rectangle (7,4.5)}
	\newcommand{\thirdMBox}{(4.25,.5) rectangle (11.5,1.75)}

	\fill[box,\lightgray] \eCapBox;
	\fill[box,\darkgray] \mCapBox;

	\begin{scope}
		\fill[thirdBoxOuter] \thirdMBox;
		\clip \eCapBox;
		\fill[thirdBox] \thirdMBox;
		\clip \mCapBox;
		\fill[thirdBoxInner] \thirdMBox;
	\end{scope}

	\begin{scope}
		\clip \eCapBox;
		\fill[box,\lightgray] \secondMBox;
		\clip \mCapBox;
		\fill[box,\darkgray] \secondMBox;
	\end{scope}
	\begin{scope}
		\fill[secondBoxOuter] \secondMBox;
		\clip \eCapBox;
		\fill[secondBox] \secondMBox;
		\clip \mCapBox;
		\clip \secondMBox;
		\fill[secondBoxInner] \secondMBox;
	\end{scope}

	\begin{scope}
		\clip \eCapBox;
		\fill[box,\lightgray] \firstMBox;
		\clip \mCapBox;
		\fill[box,\darkgray] \firstMBox;
	\end{scope}
	\begin{scope}
		\fill[firstBoxOuter] \firstMBox;
		\clip \eCapBox;
		\fill[firstBox] \firstMBox;
		\clip \mCapBox;
		\clip \firstMBox;
		\fill[firstBoxInner] \firstMBox;
	\end{scope}

	\draw[box] \thirdMBox;
	\draw[box] \secondMBox;
	\draw[box] \firstMBox;
	\draw[box] \mCapBox;
	\draw[box] \eCapBox;

	\tikzset{label/.style={rectangle,rounded corners,fill=gray!1,inner sep=1pt,fill opacity=.8,text opacity=1}}
	\contourlength{0.06em}
	\newcommand{\xOffset}{0.33}
	\newcommand{\yOffset}{-0.28}
	\node[] at($(0,3.75)+(\xOffset,\yOffset)$) {\contour{white}{$E_\cap$}};
	\node[] at ($(2,3)+(\xOffset,\yOffset)$) {\contour{white}{$M^*_\cap$}};
	\node[] at ($(-1.5,2)+(.15+\xOffset,\yOffset)$) {\contour{white}{$M^{(1)}_1$}};
	\node[] at ($(0,2)+(\xOffset,\yOffset)$) {\contour{white}{$R_1$}};
	\node[] at ($(2,2)+(\xOffset,\yOffset)$) {\contour{white}{$R^{*}_1$}};
	\node[] at ($(4,4.5)+(.15+\xOffset,\yOffset)$) {\contour{white}{$M^{(2)}_1$}};
	\node[] at ($(4,3.75)+(\xOffset,\yOffset)$) {\contour{white}{$R_2$}};
	\node[] at ($(4,3)+(\xOffset,\yOffset)$) {\contour{white}{$R^{*}_2$}};
	\node[] at ($(11.5,1.75)+(-.15-\xOffset,\yOffset)$) {\contour{white}{$M^{(3)}_1$}};
	\node[] at ($(10,1.75)+(-\xOffset,\yOffset)$) {\contour{white}{$R_3$}};
	\node[] at ($(8,1.75)+(-\xOffset,\yOffset)$) {\contour{white}{$R^{*}_3$}};
\end{tikzpicture}

%% file: tikz-belp.tex
\begin{tikzpicture}[xscale=2, yscale=1.5]
	\newcommand{\covfefe}{
		\node[dot] (a) at (0,0) {};
		\node[dot] (b) at (1,0) {};
		\node[dot] (c) at (2,0) {};
		\node[dot] (d) at (2,1) {};
		\node[dot] (e) at (1,1) {};
		\node[dot] (f) at (0,1) {};
	}

	\newcommand{\tartufo}[2]{
		\draw[R] (a) -- (b);
		\draw[R] (b) -- (c);
		\draw[R] (c) -- node[black, right]{$#2$} (d);
		\draw[R] (d) -- (e);
		\draw[R] (e) -- (f);
		\draw[R] (f) -- node[black, left]{$#1$} (a);
	}

	\newcommand{\toffifee}[1]{
		\draw[B] (a) to[bend right] (c);
		\draw[B] (c) -- (e);
		\draw[B] (e) --node[black, left]{$#1$} (b);
		\draw[B] (b) -- (d);
		\draw[B] (d) to[bend right] (f);
		\draw[B] (f) -- (a);
		\draw[B] (c) -- (d);
	}

	\path[use as bounding box](0,-.35) rectangle (5,1.35);
	\tikzset{dot/.style={draw,circle,fill=black,inner sep=1.5}}
	\tikzset{r/.style={superred,line width=.4mm}}
	\tikzset{b/.style={superblue,semithick,decorate,decoration={snake, segment length=2mm, amplitude=.5mm}}}

	\tikzset{R/.style={r}}
	\tikzset{B/.style={b}}
	\covfefe
	\tartufo{e_1}{e_2}
	\toffifee{f_1}
	\begin{scope}[xshift=3cm]
		\tikzset{R/.style={b}}
		\tikzset{B/.style={r}}
		\covfefe
		\toffifee{f_2}
		\tartufo{e_3}{e_4}
	\end{scope}
\end{tikzpicture}

%% file: tikz-hardness.tex
\begin{tikzpicture}[yscale=.8]
	\useasboundingbox[green] (-4,-1.25) rectangle (4,1.25);
	\tikzset{dot/.style={draw,circle,fill=black,inner sep=1.5}}
	\tikzset{R/.style={superred,line width=.4mm}}
	\tikzset{B/.style={superblue,semithick,decorate,decoration={snake, segment length=2mm, amplitude=.5mm}}}
	\tikzset{dots/.style={B,semithick,densely dotted}}

	\node[dot] at (-1,0) (L2) {};
	\node[dot, above of=L2] (L1) {};
	\node[dot, below of=L2] (L3) {};
	\node[dot] at (1,0) (L5) {};
	\node[dot, above of=L5] (L0) {};
	\node[dot, below of=L5] (L4) {};

	\draw (L0) edge[B] (L1);
	\draw (L1.center) edge[B] (L2.center);
	\draw (L1.center) edge[R] (L2.center);
	\draw (L2.center) edge[B] (L3.center);
	\draw (L2.center) edge[R] node[left,black] {$y^e_v$} (L3.center);
	\draw (L3) edge[B] node[black,below,yshift=-3mm] {$e=vw$} (L4);
	\draw (L4.center) edge[B] (L5.center);
	\draw (L4.center) edge[R] (L5.center);
	\draw (L5.center) edge[B] (L0.center);
	\draw (L5.center) edge[R] node[right,black] {$y^e_w$} (L0.center);

	\node[dot] [left of=L2, node distance=2cm] (Cx2) {};
	\node[dot] [below of=Cx2] (Cx1) {};
	\node[dot] [above of=Cx2] (Cx3) {};
	\node[dot] [below left of=Cx1,yshift=3mm] (Cx0) {};
	\node[dot] [above left of=Cx3,yshift=-3mm] (Cx4) {};

	\draw (Cx0) edge[B] (Cx1);
	\draw (Cx1.center) edge[B] (Cx2.center);
	\draw (Cx1.center) edge[R] node[black,right] {$x^e_v$} (Cx2.center);
	\draw (Cx2.center) edge[B] (Cx3.center);
	\draw (Cx2.center) edge[R] (Cx3.center);
	\draw (Cx3) edge[B] (Cx4);

	\node [left of=Cx4,xshift=4.5mm] (Cx5x) {};
	\node [left of=Cx4,xshift=2.5mm] (Cx5) {};
	\draw (Cx4.center) edge[B] (Cx5x);
	\draw (Cx4.center) edge[dots] (Cx5);
	\node [left of=Cx0,xshift=4.5mm] (Cx-1x) {};
	\node [left of=Cx0,xshift=2.5mm] (Cx-1) {};
	\draw (Cx0.center) edge[B] (Cx-1x);
	\draw (Cx0.center) edge[dots] (Cx-1);

	\draw (Cx3) edge[R] (L1);
	\draw (L3) edge[R] (Cx1);

	\node[dot] [right of=L5, node distance=2cm] (Cy2) {};
	\node[dot] [below of=Cy2] (Cy1) {};
	\node[dot] [above of=Cy2] (Cy3) {};
	\node[dot] [below right of=Cy1,yshift=3mm] (Cy0) {};
	\node[dot] [above right of=Cy3,yshift=-3mm] (Cy4) {};

	\draw (Cy0) edge[B] (Cy1);
	\draw (Cy1.center) edge[B] (Cy2.center);
	\draw (Cy1.center) edge[R] (Cy2.center);
	\draw (Cy2.center) edge[B] (Cy3.center);
	\draw (Cy3) edge[B] (Cy4);
	\draw (Cy2.center) edge[R] node[black,left] {$x^e_w$} (Cy3.center);

	\node [right of=Cy4,xshift=-4.5mm] (Cy5x) {};
	\node [right of=Cy4,xshift=-2.5mm] (Cy5) {};
	\draw (Cy4.center) edge[B] (Cy5x);
	\draw (Cy4.center) edge[dots] (Cy5);
	\node [right of=Cy0,xshift=-4.5mm] (Cy-1x) {};
	\node [right of=Cy0,xshift=-2.5mm] (Cy-1) {};
	\draw (Cy0.center) edge[B] (Cy-1x);
	\draw (Cy0.center) edge[dots] (Cy-1);

	\draw (Cy3) edge[R] (L0);
	\draw (L4) edge[R] (Cy1);

	\draw[decorate,decoration=brace,xshift=1mm] (-0.9,0 |- L1.south)
	-- node[right] {\small $Y^e_v$} (-0.9,0 |- L3.north);
	\draw[decorate,decoration=brace,xshift=-1mm] (0.9,0 |- L4.north)
	-- node[left] {\small $Y^e_w$} (0.9,0 |- L0.south);
	\draw[decorate,decoration=brace,xshift=-1mm] (-3.1,0 |- Cx1.north)
	-- node[left] {\small $X^e_v$} (-3.1,0 |- Cx3.south);
	\draw[decorate,decoration=brace,xshift=1mm] (3.1,0 |- Cy3.center)
	-- node[right] {\small $X_w^e$} (3.1,0 |- Cy1.north);

	\foreach \x in {0,...,4}{
		\node[dot] at (L\x) {};
		\node[dot] at (Cx\x) {};
		\node[dot] at (Cy\x) {};
	}
	\node[dot] at (L5) {};
	\node[draw,circle,semithick,fill=none,minimum width=3mm,inner sep=0] at (L0) {};
	\node[draw,circle,semithick,fill=none,minimum width=3mm,inner sep=0] at (L3) {};
	\node[draw,circle,semithick,fill=none,minimum width=3mm,inner sep=0] at (Cx1) {};
	\node[draw,circle,semithick,fill=none,minimum width=3mm,inner sep=0] at (Cy3) {};
\end{tikzpicture}

%% file: tikz-lpgap.tex
\begin{tikzpicture}[scale=.9, xscale=.9, yscale=.9]
	\pgfmathsetmacro{\ksize}{3}
	\pgfmathsetmacro{\rEnd}{4*\ksize+1}
	\path[use as bounding box](-1,-.25-\ksize) rectangle (\rEnd+1,.25+\ksize);
	\tikzset{every node/.style={draw,circle,fill=black,inner sep=1.5}}
	\tikzset{r/.style={superred,line width=.4mm}}
	\tikzset{b/.style={superblue,semithick,decorate,decoration={snake, segment length=2mm, amplitude=.7mm}}}
	\tikzset{g/.style={green!50!black,very thick,densely dotted}}

	\foreach \i in {0,...,\ksize}{
		\foreach \j in {0,...,\ksize}{
			\pgfmathsetmacro{\xcoor}{int(2*(\i+\j))}
			\pgfmathsetmacro{\ycoor}{int(\i-\j)}
			\node (a\i\j) at (\xcoor,\ycoor) {};
			\node[xshift=1cm] at (a\i\j) (b\i\j) {};
			\draw[b] (a\i\j) -- (b\i\j);
			\draw[r] (a\i\j) -- (b\i\j);
		}
	}

	\foreach \i in {1,...,\ksize}{
		\foreach \j in {0,...,\ksize}{
			\pgfmathsetmacro{\iLess}{int(\i-1)}
			\draw[r] (a\i\j) -- (b\iLess\j);
		}
	}

	\foreach \i in {0,...,\ksize}{
		\foreach \j in {1,...,\ksize}{
			\pgfmathsetmacro{\jLess}{int(\j-1)}
			\draw[b] (a\i\j) -- (b\i\jLess);
		}
	}

	\foreach \j in {1,...,\ksize}{
		\pgfmathsetmacro{\jLess}{int(\j-1)}
		\node[xshift=-1cm] (c\j) at (a0\j) {};
		\draw[r] (c\j) -- (a0\jLess);
		\draw[r] (c\j) -- (a0\j);

		\node[xshift=1cm] (d\j) at (b\ksize\jLess) {};
		\draw[r] (d\j) -- (b\ksize\jLess);
		\draw[r] (d\j) -- (b\ksize\j);

		\pgfmathsetmacro{\phi}{int(\ksize-\j+1)}
		\node[xshift=-1cm] (cc\phi) at (a\j0) {};
		\draw[b] (cc\phi) -- (a\jLess0);
		\draw[b] (cc\phi) -- (a\j0);

		\node[xshift=1cm] (dd\phi) at (b\jLess\ksize) {};
		\draw[b] (dd\phi) -- (b\jLess\ksize);
		\draw[b] (dd\phi) -- (b\j\ksize);
	}

	\draw[g] (c1) .. controls (-2.5,0) and (-2.5,\ksize) .. (cc1);
	\draw[g] (c2) .. controls (-2.5,-1.5) and (-2.5,1.5) .. (cc2);
	\draw[g] (c3) .. controls (-2.5,-\ksize) and (-2.5,0) .. (cc3);

	\draw[g] (dd1) .. controls (\rEnd+2.5,0) and (\rEnd+2.5,\ksize) .. (d1);
	\draw[g] (dd2) .. controls (\rEnd+2.5,-1.5) and (\rEnd+2.5,1.5) .. (d2);
	\draw[g] (dd3) .. controls (\rEnd+2.5,-\ksize) and (\rEnd+2.5,0) .. (d3);

	\tikzset{label/.style={draw=none,fill=none}}
	\tikzset{rightLabel/.style={label,xshift=13pt}}
	\tikzset{leftLabel/.style={label,xshift=-12pt}}
	\tikzset{upLabel/.style={label,yshift=7pt}}
	\tikzset{downLabel/.style={label,yshift=-7pt}}

	\tikzset{leftLabeldown/.style={leftLabel,xshift=6pt,yshift=-6pt}}
	\tikzset{leftLabelup/.style={leftLabel,xshift=6pt,yshift=6pt}}
	\tikzset{rightLabeldown/.style={rightLabel,xshift=-6pt,yshift=-7pt}}
	\tikzset{rightLabelup/.style={rightLabel,xshift=-6pt,yshift=6pt}}
	\node[leftLabel,xshift=2pt] at (a00) {$a_{00}$};
	\node[leftLabeldown] at (a01) {$a_{01}$};
	\node[leftLabeldown] at (a02) {$a_{02}$};
	\node[downLabel] at (a03) {$a_{03}$};

	\node[leftLabelup] at (a10) {$a_{10}$};
	\node[leftLabel] at (a11) {$a_{11}$};
	\node[leftLabel] at (a12) {$a_{12}$};
	\node[leftLabel] at (a13) {$a_{13}$};

	\node[leftLabelup] at (a20) {$a_{20}$};
	\node[leftLabel] at (a21) {$a_{21}$};
	\node[leftLabel] at (a22) {$a_{22}$};
	\node[leftLabel] at (a23) {$a_{23}$};

	\node[upLabel] at (a30) {$a_{30}$};
	\node[leftLabel] at (a31) {$a_{31}$};
	\node[leftLabel] at (a32) {$a_{32}$};
	\node[leftLabel] at (a33) {$a_{33}$};

	\node[rightLabel] at (b00) {$b_{00}$};
	\node[rightLabel] at (b01) {$b_{01}$};
	\node[rightLabel] at (b02) {$b_{02}$};
	\node[downLabel] at (b03) {$b_{03}$};

	\node[rightLabel] at (b10) {$b_{10}$};
	\node[rightLabel] at (b11) {$b_{11}$};
	\node[rightLabel] at (b12) {$b_{12}$};
	\node[rightLabeldown] at (b13) {$b_{13}$};

	\node[rightLabel] at (b20) {$b_{20}$};
	\node[rightLabel] at (b21) {$b_{21}$};
	\node[rightLabel] at (b22) {$b_{22}$};
	\node[rightLabeldown] at (b23) {$b_{23}$};

	\node[upLabel] at (b30) {$b_{30}$};
	\node[rightLabelup] at (b31) {$b_{31}$};
	\node[rightLabelup] at (b32) {$b_{32}$};
	\node[rightLabel,xshift=-2pt,yshift=2pt] at (b33) {$b_{33}$};

	\node[downLabel] at (c1) {$c_{1}$};
	\node[downLabel] at (c2) {$c_{2}$};
	\node[downLabel] at (c3) {$c_{3}$};
	\node[upLabel] at (d1) {$d_{1}$};
	\node[upLabel] at (d2) {$d_{2}$};
	\node[upLabel] at (d3) {$d_{3}$};

	\node[upLabel] at (cc1) {$c_{1}$};
	\node[upLabel] at (cc2) {$c_{2}$};
	\node[upLabel] at (cc3) {$c_{3}$};
	\node[downLabel] at (dd1) {$d_{1}$};
	\node[downLabel] at (dd2) {$d_{2}$};
	\node[downLabel] at (dd3) {$d_{3}$};
\end{tikzpicture}

%% file: tikz-tightness.tex
\begin{tikzpicture}[xscale=1.45,yscale=.9]
	\tikzset{every node/.style={draw,circle,fill=black,inner sep=1.5,text height=0pt}}
	\tikzset{R/.style={superred,line width=.5mm}}
	\tikzset{B/.style={superblue,line width=.25mm,decorate,decoration={snake, segment length=1.5mm, amplitude=.5mm}}}

	\tikzset{label/.style={draw=none,fill=none,text height=5pt}}
	\tikzset{leftLabel/.style={label,shift={(-.35,0)}}}
	\tikzset{rightLabel/.style={label,shift={(.35,0)}}}
	\tikzset{upLabel/.style={label,yshift=8pt}}
	\tikzset{downLabel/.style={label,yshift=-7pt}}

	\newcommand{\xShift}{.2}
	\newcommand{\yShift}{.2}
	\tikzset{leftLabeldown/.style={label,shift={(-\xShift,-\yShift)}}}
	\tikzset{leftLabelup/.style={label,shift={(-\xShift,\yShift)}}}
	\tikzset{rightLabeldown/.style={label,shift={(\xShift,-\yShift)}}}
	\tikzset{rightLabelup/.style={label,shift={(\xShift,\yShift)}}}

	\set{\k}{4}
	\set{\kMinus}{\k-1}
	\set{\kMinusTwo}{\k-2}
	\set{\kDouble}{2*\k}
	\set{\kHalf}{\k/2}
	\set{\kHalfMinus}{\kHalf -1}
	\set{\kHalfPlus}{\kHalf +1}

	\newcommand{\cycleHeight}{.75}
	\newcommand{\cyclewidth}{.25}
	\foreach \i in {1,...,\k}{
		\set{\xa}{2*\i-2}
		\set{\xb}{2*\i-1}
		\node (w'\i) at (\xa,0) {};
		\node (w\i)  at (\xb,0) {};
		\node[leftLabel] at (w'\i) {$w'_\i$};
		\node[rightLabel] at (w\i) {$w_\i$};

		\set{\iMinus}{\i-1}
		\draw[R] (w'\i) -- (w\i);
		\draw[B] (w'\i) -- (w\i);

		\node (w''\i)  at (\xa+\cyclewidth,\cycleHeight) {};
		\node (w'''\i) at (\xb-\cyclewidth,\cycleHeight) {};
		\draw[B] (w'\i) -- (w''\i);
		\draw[B] (w''\i) -- (w'''\i);
		\draw[B] (w'''\i) -- (w\i);

		\node (W''\i)  at (\xa+\cyclewidth,-\cycleHeight) {};
		\node (W'''\i) at (\xb-\cyclewidth,-\cycleHeight) {};
		\draw[R] (w'\i) -- (W''\i);
		\draw[R] (W''\i) -- (W'''\i);
		\draw[R] (W'''\i) -- (w\i);
	}

	\newcommand{\detourHeight}{.5}
	\newcommand{\detourwidth}{.2}
	\foreach \i in {1,...,\kMinus}{
		\set{\xa}{2*\i-1}
		\set{\xb}{2*\i}
		\set{\iPlus}{\i+1}
		\node (W''\i)  at (\xa+\detourwidth,\detourHeight) {};
		\node (W'''\i) at (\xb-\detourwidth,\detourHeight) {};
		\draw[R] (w\i) -- (W''\i);
		\draw[R] (W''\i) -- (W'''\i);
		\draw[R] (W'''\i) -- (w'\iPlus);

		\node (q''\i)  at (\xa+\detourwidth,-\detourHeight) {};
		\node (q'''\i) at (\xb-\detourwidth,-\detourHeight) {};
		\draw[B] (w\i) -- (q''\i);
		\draw[B] (q''\i) -- (q'''\i);
		\draw[B] (q'''\i) -- (w'\iPlus);
	}

	\newcommand{\offsetOne}{\cycleHeight-.5}
	\foreach \i in {1,...,\kMinus}{
		\set{\xa}{2*\i-1}
		\set{\xb}{2*\i}
		\set{\iPlus}{\i+1}
		\node (v\i) at (\xa,-\offsetOne) {};
		\node[leftLabeldown] at (v\i) {$v_\i$};
		\node (phiv\i) at (\xb,-\offsetOne) {};
		\draw[B] (w\i) -- (v\i);
		\draw[R] (w'\iPlus) -- (phiv\i);
	}
	\node (u0) at (0,-\offsetOne) {};
	\node[rightLabel] at (u0) {$u_0$};
	\node (phiu0) at (\kDouble-1,-\offsetOne) {};
	\draw[B] (u0) -- (w'1);
	\draw[R] (phiu0) -- (w4);

	\foreach \i in {1,...,\kMinus}{
		\foreach \j in {\i,...,\kMinus}{
				\node (u) at (0,-\offsetOne-\k+.5) {};
				\node (u0) at (0,-\offsetOne) {};
			\set{\xa}{\i+\j-1}
			\set{\xb}{\i+\j}
			\set{\ya}{\i-\j-1}
			\node (a\i\j) at (\xa+.25,\ya-\offsetOne+.25) {};
			\node (b\i\j) at (\xb-.25,\ya-\offsetOne+.25) {};
			\node[leftLabel,shift={(0,-.03)}] at (a\i\j) {$a^\j_\i$};
			\node[rightLabel,shift={(0,-.05)}] at (b\i\j) {$b^\j_\i$};
			\draw[R] (a\i\j) -- (b\i\j);
			\draw[B] (a\i\j) -- (b\i\j);
		}
	}
	\foreach \i in {1,...,\kMinus}{
		\draw[B] (v\i) to[bend left=10] (b\i\i);
		\draw[R] (phiv\i) to[bend right=10] (a\i\i);
		\set{\phii}{\k-\i}
		\node (u\i) at (\i,-\i-.5-\offsetOne) {};
		\node[leftLabelup] at (u\i) {$u_\i$};
		\node (phiu\phii) at (\k+\phii-1,-\i-.5-\offsetOne) {};
		\draw[B] (u\i) -- (a1\i);
		\draw[R] (phiu\phii) -- (b\phii\kMinus);
	}

	\foreach \j in {2,...,\kMinus}{
		\set{\jMinus}{\j-1}
		\foreach \i in {1,...,\jMinus}{
			\set{\iPlus}{\i+1}
			\set{\iPhi}{\k-\j}
			\set{\iPhiPlus}{\k-\j+1}
			\set{\jPhi}{\k-\i}
			\draw[B] (b\i\j) -- (a\iPlus\j);
			\draw[R] (a\i\j) -- (b\i\jMinus);
		}
	}

	\node (u) at (0,-\kMinus-.5-\offsetOne) {};
	\node[leftLabel,shift={(.1,0)}] at (u) {$u$};
	\node (phiu) at (\kDouble-1,-\kMinus-.5-\offsetOne) {};
	\foreach \i in {0,...,\kMinus}{
		\draw[B] (u) -- (u\i);
		\draw[R] (phiu) -- (phiu\i);
	}
	\node[label] (c1) at ($(phiu)+(1.2,-1.2)$) {};
	\node[label] (c2) at ($(u)+(-1.2,-1.2)$) {};
	\draw[B] (u) .. controls (c1) .. (w\k);
	\draw[R] (phiu) .. controls (c2) .. (w'1);

	\pgfresetboundingbox
	\path[green,use as bounding box] ($(u)+(-.7,-.3)$) rectangle ($(w\k)+(.7,.75)$);
\end{tikzpicture}